\documentclass[letterpaper]{scrartcl} 

\usepackage[]{brandl}
\usepackage{bm,esvect,tabu}
\usepackage{tabularx}

\theoremstyle{definition}
\newtheorem{swf}{Function}

\newcolumntype{Y}{>{\centering\arraybackslash}X}

\newcommand{\citeauthorhloff}[1]{{\protect\NoHyper\citeauthor{#1}\protect\endNoHyper}}

\newcommand{\sgn}{\mathrm{sign}}
\DeclareMathOperator{\ev}{\mathbb{E}}% 

\newcommand{\baru}{\ensuremath{\mathrm{BARU}}\xspace}

\newcommand{\p}[1][]{% 
\ifthenelse{\equal{#1}{}}{{\bm\succcurlyeq}}{{\bm\succcurlyeq_{#1}}}% 
}

\newcommand{\s}[1][]{% 
\ifthenelse{\equal{#1}{}}{{\succcurlyeq}}{{\succcurlyeq_{#1}}}% 
}

\newcommand{\srel}[1][]{% 
\ifthenelse{\equal{#1}{}}{\succcurlyeq}{\succcurlyeq_{#1}}% 
}

\title{Belief-Averaging and Relative Utilitarianism}

\author{Florian Brandl\thanks{University of Bonn, Germany, \texttt{florian.brandl@uni-bonn.de}}\\[1ex]
\large (\href{http://brandlf.com/docs/baru.pdf}{click for most recent revision})}

\date{}

\begin{document}
	
\maketitle

\begin{abstract}
	We consider social welfare functions when the preferences of individual agents and society maximize subjective expected utility in the tradition of Savage. A system of axioms is introduced whose unique solution is the social welfare function that averages the agents’ beliefs and sums up their utility functions, normalized to have the same range. The first distinguishing axiom requires positive association of society's preferences with the agents' preferences for acts about which beliefs agree. The second is a weakening of Arrow’s independence of irrelevant alternatives that only applies to non-redundant acts.
\par\vskip\baselineskip
\end{abstract}

\citet{Hars55a} studied how a society should rank risky alternatives.
His aggregation theorem shows that if the agents as well as society are von Neumann-Morgenstern expected utility maximizers and the preferences of society satisfy the Pareto principle, then the utility function of society is a linear combination of the agents' utility functions.
Among others, two streams of literature originated from here.

The first studies the implications of the Pareto principle when the social alternatives are not lotteries but uncertain acts.
The approach of \citet{Sava54a} is to view acts as maps from states of the world to outcomes.
A preference relation over acts maximizes subjective expected utility if it ranks acts by their expected utility for a utility function over outcomes and an idiosyncratic belief that assigns probabilities to states.
\citet{Mong95a} showed that if agents and society are \emph{subjective} expected utility maximizers, the Pareto principle implies that the preferences of society have to coincide with those of one of the agents unless all agents hold the same belief (in which case we are back to \citeauthorhloff{Hars55a}'s setting) or the same utility function.
\citeauthorhloff{Mong95a} noted, however, that the Pareto principle also applies to \emph{spurious unanimities} where all agents prefer one act to another because of differences in beliefs \emph{and} differences in utility functions.
These differences can cancel out each other and lead to a unanimous preference with agents having incompatible reasons for preferring the first act. 
\citet{Mong97a} argues that the Pareto principle in its full force is thus not compelling when applied to subjective expected utility maximizers.
\citet*{GSS04a} introduced a restricted Pareto condition, which avoids spurious unanimities.
It applies only to acts that induce the same distribution over outcomes for all agents' beliefs.
Such acts correspond to lotteries in \citeauthorhloff{Hars55a}'s setting.
Thus, by \citeauthorhloff{Hars55a}'s aggregation theorem, the restricted Pareto condition implies that the utility function of society is a linear combination of the agents' utility functions.
Surprisingly, it also implies that the belief of society is an affine combination of the agents' beliefs.

The second stream stays with decisions under risk and deals with a question left open by \citeauthorhloff{Hars55a}'s work: if we follow \citeauthorhloff{Hars55a}'s utilitarian doctrine, how do we assign weights to the agents' utility functions?
The approach pursued by \citet{DhMe99a} considers social welfare functions in the tradition of \citet{Arro51a}, which map every preference profile to a preference relation for society.
\citeauthorhloff{DhMe99a} impose a set of axioms on social welfare functions that, roughly speaking, constitute weakenings of the conditions in \citeauthorhloff{Arro51a}'s impossibility theorem.
They show that only \emph{relative utilitarianism} satisfies all of their axioms: normalize the agents' utility functions so that their range is the unit interval and sum them up.\footnote{Several other characterizations of relative utilitarianism, for example, by \citet{Karn98a,Dhil98a,Sega00a,BoCh17a}, exist. 
The discussion focuses on the characterization of \citet{DhMe99a} since it is most relevant to the present paper.}
Their result is multi-profile since the axioms relate the preferences of society for different preference profiles to each other.
By contrast, the results of \citeauthorhloff{Hars55a}, \citeauthorhloff{Mong95a}, and \citeauthorhloff{GSS04a} discussed above are single-profile since Pareto conditions talk about a single preference profile considered in isolation.

This paper seeks to combine both streams.
That is, it takes a multi-profile approach to decisions under uncertainty.
We consider social welfare functions when the preferences of the agents as well as the society are subjective expected utility-maximizing.
The goal is to determine a social welfare function that can be justified on axiomatic grounds. 
To this end, two axioms are introduced.
First, assume that society is indifferent between two acts and one of the agents is completely indifferent (that is, indifferent between all acts). 
Suppose this agent changes its preferences so that its new belief induces the same distribution over outcomes as the belief underlying society's preferences before the change for either of the two acts.
Then, after the change, society's preference over the two acts should coincide with the focal agent's new preference.
This axiom, called \emph{restricted monotonicity}, captures the idea of positive association between the agents' preferences and society's preferences: if an agent changes its preferences in favor of an act, society's preferences should change likewise.
More precisely, restricted monotonicity applies the restricted Pareto condition of \citet{GSS04a} to society's preferences before the change and the focal agent's new preferences.
 
Second, recall that {\protect\NoHyper\cites{Arro51a}\protect\endNoHyper} independence of irrelevant alternatives requires that society's preference over any pair of alternatives---acts in our case---depends only on the agents' preferences over these two acts.
In the present setting, it forces the social welfare function to ignore the agents' expected utilities for the two acts apart from how they are ordered.
The second axiom weakens \citeauthorhloff{Arro51a}'s condition so that society's preference may also depend on the agents' expected utilities for the two acts.
(In fact, it is even weaker.)
Referring to \citeauthorhloff{DhMe99a}'s independence of redundant alternatives axiom by which it is inspired, it is called \emph{independence of redundant acts}.

The main result of this paper shows that these two conditions, together with four undiscriminating axioms, characterize \emph{belief-averaging and relative utilitarianism}: average the agents' beliefs and sum up their utility functions, normalized to the unit interval.
The resulting belief and utility function determine society's preferences.
This extends \citeauthorhloff{GSS04a}'s linear aggregation result to a multi-profile framework and \citeauthorhloff{DhMe99a}'s relative utilitarianism to subjective expected utility maximizers.
\Cref{tab:relatedresults} summarizes how the present paper relates to the described works.

\begin{table}
	\centering
	\begin{tabularx}{.9\textwidth}{lYY}
		& Risk & Uncertainty\\
		\midrule
			Single-profile & linear aggregation of utilities \citep{Hars55a} & linear aggregation of beliefs and utilities \citep{GSS04a}\\
			Multi-profile & relative utilitarianism \citep{DhMe99a} & belief-averaging and relative utilitarianism\\
	\end{tabularx}
	\caption{Placement of the present paper relative to previous works. The papers in the ``Risk'' (``Uncertainty'') column assume that agents as well as society are (subjective) expected utility maximizers. The single-profile results use only the (restricted) Pareto condition. The multi-profile results rely on several axioms, some of which relate the preferences of society for different preference profiles to each other.}
	\label{tab:relatedresults}
\end{table}

\paragraph{Discussion of the axioms}

Restricted monotonicity is defined as follows. 
Suppose we apply a social welfare function to some preference profile  and the thus-derived preferences of society rank two acts $f$ and $g$ as equally desirable.
Restricted monotonicity applies when a completely indifferent agent (if one exists) changes its preferences. 
If (i) this agent prefers $f$ to $g$ after changing preferences and (ii) for either of $f$ and $g$, the belief underlying the agent's new preferences induces the same distribution over outcomes as the belief underlying society's preferences for the original profile, then society should prefer $f$ to $g$ as well after the preference change.\footnote{Two beliefs induce the same distribution over outcomes for an act $f$ if the push-forward measure on the set of outcomes under $f$ is the same for both beliefs. Equivalently, both beliefs agree on the sigma-algebra over the set of states induced by $f$.}
 
Part (ii) above makes restricted monotonicity apply only to acts that are risky alternatives for society's belief for the original profile and the new belief of the previously indifferent agent.
That is, acts for which differences in those two beliefs are irrelevant and the utility functions are the only source of preference heterogeneity. 
The rationale for part (ii) is the same as that of \citet{GSS04a} for restricting the Pareto condition to acts that induce the same distribution over outcomes for all agents' beliefs: it avoids accidental preference agreements through differences in beliefs and utility functions.
To see why restricted monotonicity is less substantiated if differences in beliefs matter, consider the following example.

\begin{example}\label{ex:1}
	There are two states of the world $\omega_1$ and $\omega_2$, and three possible outcomes $a$, $b$, and $c$. 
	Call $f$ the act that results in $a$ in the state $\omega_1$ and in $b$ in the state $\omega_2$ and $g$ the act that gives outcome $c$ for both states.
	Consider two agents, Agent 1 and Agent 2, were the belief and utility function of Agent 1 are as in \Cref{tab:example} and Agent 2 is completely indifferent.\footnote{The structure of this example is very similar to the example of \citet{GSS04a} of the two gentlemen who are contemplating having a duel.} 
	Calculating expected utilities shows that Agent 1 is indifferent between $f$ and $g$. 
	Thus, let us assume that society is indifferent between $f$ and $g$ in this case.
	An unrestricted monotonicity condition (without part (ii)) would then demand that society prefers $f$ to $g$ if Agent 2 changed it preferences to those in \Cref{tab:example}.
	This conclusion is questionable.
		Agent~1 has a high expected utility for $f$ in anticipation of $f$ giving outcome $a$ and a high utility for $a$.
		Agent~2 has a high expected utility for $f$ because of a high probability on $f$ resulting in outcome $b$ and a high utility for $b$. 
		So both, their probability rankings and their utility rankings are reversed.
		Only because these differences cancel out do they both arrive at a high expected utility for $f$.
		In other words, they have different and incompatible reasons for preferring $f$.\footnote{For example, there is no belief (possibly that of an impartial observer) so that if both agents held that belief, both of them would prefer $f$ to $g$. In fact, for any belief with a probability of more than 20\% for $\omega_1$ and more than 10\% for $\omega_2$, both agents would prefer $g$ to $f$.}
		It is thus dubious to insist on monotonicity in this example.
\end{example}

\begin{table}
	\centering
	\begin{tabular}{lccccccc}
	\cmidrule[\heavyrulewidth]{1-8}
	& \multicolumn{2}{c}{belief} & \multicolumn{3}{c}{utility function} & \multicolumn{2}{c}{expected utilities}\\
	& $\omega_1$ & $\omega_2$ & $a$ & $b$ & $c$ & $f$ & $g$\\
	\cmidrule{1-8}
	Agent 1 & 90\% & 10\% & 1 & 0 & 0.9 & 0.9 & 0.9\\
	Agent 2 & 10\% & 90\% & 0 & 1 & 0.8 & 0.9 & 0.8\\
	\cmidrule[\heavyrulewidth]{1-8}
	\end{tabular}
	\caption{Numerical values for \Cref{ex:1}.
	$\omega_1$ and $\omega_2$ are the two states of the world and $a$, $b$, and $c$ are possible outcomes.
	$f$ is an act that yields $a$ if the state is $\omega_1$ and $b$ if the state is $\omega_2$; $g$ yields $c$ in both states.
	The last two rows give the beliefs, utility functions, and expected utilities of Agent 1 and Agent 2.}
	\label{tab:example}
\end{table}

This is not to say that any monotonicity condition stronger than the one defined above is undesirable. 
The approach here is merely to be cautious about the assumptions that are imposed.
Note however that when requiring that society is indifferent between all acts if all agents are (which we call faithfulness), an unrestricted monotonicity condition implies the (unrestricted) Pareto principle and falls prey to the impossibility of \citet{Mong97a}.

The restricted monotonicity axiom and faithfulness together imply the restricted Pareto condition of \citeauthorhloff{GSS04a}.
Hence, the belief and utility function of the society have to be linear combinations of the agents' beliefs and utility functions.
The restricted Pareto condition allows the weights in both linear combinations to be arbitrary functions of \emph{all} agents' preferences.
The additional strength of restricted monotonicity implies that the weight of an agent in either linear combination depends only on the agent's own preferences.

{\protect\NoHyper\cites{Arro51a}\protect\endNoHyper} independence of irrelevant alternatives prescribes that the preferences of society over any two acts must depend only on the agents' preferences over the two acts, and not on their preferences over other acts.
It precludes that society's ranking of two acts depends more flexibly on the agents' expected utilities for the two acts.
For example, consider two agents whose expected utilities for three acts $f$, $g$, and $h$ are $1$, $0$, and $\alpha$ (for the first agent) and $0$, $1$, and $\alpha$ (for the second agent). 
Independence of irrelevant alternatives asserts that the preferences of this two-agent society over the three acts are the same for every value of $\alpha$ strictly between $0$ and $1$.
However, it does not seem unreasonable that $h$ is society's least preferred act if $\alpha$ is close to $0$ and the most preferred act if $\alpha$ is close to $1$.
This is ruled out by independence of irrelevant alternatives. 
Independence of redundant acts weakens independence of irrelevant alternatives so that it does not apply to profiles with different values for~$\alpha$.

Since the agents report preference relations that identify utility functions only up to scaling, we need to capture the above intuition on the level of preferences.
Call a set of acts $\mathcal A'$ \emph{co-redundant} for a profile if every act is unanimously indifferent to some act in $\mathcal A'$.\footnote{In the formal definition in \Cref{sec:axioms}, we will additionally require that each agent's preferences over $\mathcal A'$ are subjective expected utility-maximizing for a \emph{non-atomic} belief.}
Every act outside a co-redundant set $\mathcal A'$ is redundant in the sense that $\mathcal A'$ contains an act that all agents consider equally desirable.
Independence of redundant acts demands that if $\mathcal A'$ is co-redundant for two profiles and every agent has the same preferences over the acts in $\mathcal A'$ in both profiles, then society's preferences over $\mathcal A'$ are the same for both profiles.
Thus, society's preferences over the co-redundant acts must not depend on the agents' preferences over the redundant acts.
This condition restricts independence of irrelevant alternatives to pairs of profiles that are equal except for preferences over redundant acts.
Independence of redundant acts extends {\protect\NoHyper\cites{DhMe99a}\protect\endNoHyper} independence of redundant alternatives to decisions under uncertainty.

It is convenient to consider the situation in utility space.
A preference profile induces a map from acts to $\mathbb R^n$ (when $n$ is the number of agents) by mapping each act to the vector of expected utilities with one component for every agent.
A set of acts is co-redundant for a profile if and only if removing all other acts does not change the image of this map.
Hence, for two profiles that agree on a co-redundant set, this map has the same image and every act in the co-redundant set maps to the same point.
Independence of redundant acts requires that society ranks the co-redundant acts in the same order for both profiles.
Thus, society's preferences over the co-redundant acts can only depend on the agents' expected utilities and not on other features of the profiles.
The rationale for this condition is that the agents' utilities should be traded off against each other in the same way across profiles.
In this sense, independence of redundant acts is a consistency requirement on social welfare functions.
\Cref{fig:ira} illustrates the preceding discussion.

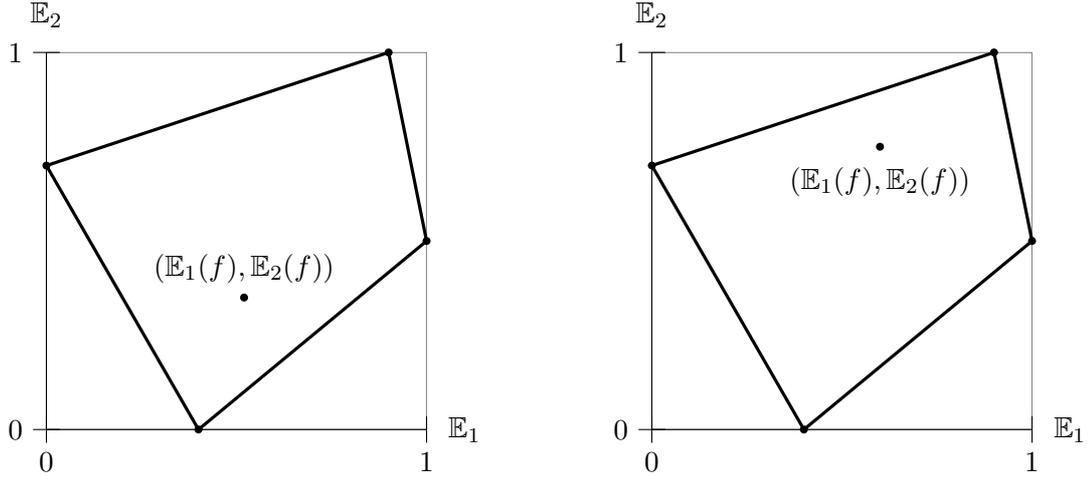
\begin{figure}[t]
			\centering
			\hfill
			\begin{tikzpicture}[scale=5]
			  \draw[gray,very thin,step=1] (0,0) grid (1,1);
			 
			  \draw (0,0) -- (1,0);
			  \foreach \x/\xtext in {0/0, 1/1}
			    \draw[shift={(\x,0)}] (0pt,1pt) -- (0pt,-1pt) node[below] {$\xtext$};
				 \node at (1.10,0) {$\ev_1$};
			 
			  \draw (0,0) -- (0,1);
			  \foreach \y/\ytext in {0/0, 1/1}
			    \draw[shift={(0,\y)}] (1pt,0pt) -- (-1pt,0pt) node[left] {$\ytext$};
				 
				\node at (0,1.1) {$\ev_2$};

			  \node[circle, fill=black, inner sep=0pt, minimum size=3pt, label = below:] at (.4,0) {};
			  \node[circle, fill=black, inner sep=0pt, minimum size=3pt, label = above:] at (.9,1) {};
			  \node[circle, fill=black, inner sep=0pt, minimum size=3pt, label = left:] at (0,.7) {};
			  \node[circle, fill=black, inner sep=0pt, minimum size=3pt, label = left:] at (1,.5) {};
			  \node[circle, fill=black, inner sep=0pt, minimum size=3pt, label={[label distance=0cm]90:{$(\ev_1(f),\ev_2(f))$}}] at (.52,.35) {};

			  \draw[shorten >=-.15pt,very thick,triangle 90 cap-triangle 90 cap] (.4,0) -- (1,.5);
			  \draw[shorten >=-.15pt,very thick,triangle 90 cap-triangle 90 cap] (1,.5) -- (.9,1);
			  \draw[shorten >=-.15pt,very thick] (.4,0) -- (0,.7);
			  \draw[shorten >=-.15pt,very thick,triangle 90 cap-triangle 90 cap] (0,.7) -- (.9,1);
			\end{tikzpicture}
			\hfill
			\begin{tikzpicture}[scale=5]
			  \draw[gray,very thin,step=1] (0,0) grid (1,1);
			 
			  \draw (0,0) -- (1,0);
			  \foreach \x/\xtext in {0/0, 1/1}
			    \draw[shift={(\x,0)}] (0pt,1pt) -- (0pt,-1pt) node[below] {$\xtext$};
				 \node at (1.10,0) {$\ev_1$};
			 
			  \draw (0,0) -- (0,1);
			  \foreach \y/\ytext in {0/0, 1/1}
			    \draw[shift={(0,\y)}] (1pt,0pt) -- (-1pt,0pt) node[left] {$\ytext$};
				 
				\node at (0,1.1) {$\ev_2$};

			  \node[circle, fill=black, inner sep=0pt, minimum size=3pt, label = below:] at (.4,0) {};
			  \node[circle, fill=black, inner sep=0pt, minimum size=3pt, label = above:] at (.9,1) {};
			  \node[circle, fill=black, inner sep=0pt, minimum size=3pt, label = left:] at (0,.7) {};
			  \node[circle, fill=black, inner sep=0pt, minimum size=3pt, label = left:] at (1,.5) {};
			  \node[circle, fill=black, inner sep=0pt, minimum size=3pt, label={[xshift=0cm, yshift=-0.85cm]:{$(\ev_1(f),\ev_2(f))$}}] at (.6,.75) {};

			  \draw[shorten >=-.15pt,very thick,triangle 90 cap-triangle 90 cap] (.4,0) -- (1,.5);
			  \draw[shorten >=-.15pt,very thick,triangle 90 cap-triangle 90 cap] (1,.5) -- (.9,1);
			  \draw[shorten >=-.15pt,very thick] (.4,0) -- (0,.7);
			  \draw[shorten >=-.15pt,very thick,triangle 90 cap-triangle 90 cap] (0,.7) -- (.9,1);

			\end{tikzpicture}
			\hfill
	\caption{Illustration of independence of redundant acts. 
	Given the preferences of two agents, each act corresponds to a point in $\mathbb R^2$ with the horizontal (vertical) coordinate equal to the expected utility of the first (second) agent.
	Assume the image of all acts under this (in general non-injective) map is the indicated quadrangle on the left when normalizing the utilities to the unit interval.
	A set of acts $\mathcal A'$ is co-redundant if its image is the entire quadrangle since then each act is unanimously indifferent to an act in $\mathcal A'$.
	If the agents change their preferences over acts that are not in $\mathcal A'$ so that still every act gets mapped to a point in the quadrangle, then $\mathcal A'$ is also co-redundant for the new preference profile.
	Independence of redundant acts implies that the preferences of society over $\mathcal A'$ must not change. 
	Such a preference change is shown on the right for an act $f$.}
	\label{fig:ira}
\end{figure}

On top of restricted monotonicity and independence of redundant acts, we make four assumptions about the social welfare function: if all agents are completely indifferent, so is society (\emph{faithfulness}), no agent can impose its belief on society (\emph{no belief imposition}), the society's preferences depend continuously on the agents' preferences (\emph{continuity}), and all agents are treated symmetrically (\emph{anonymity}).  
Together these six conditions characterize belief-averaging and relative utilitarianism.
The conceptual novelty of the paper is the restricted monotonicity condition. 
Our result shows that it is the right axiom to add to \citeauthorhloff{DhMe99a}'s independence of redundant alternatives (acts) to derive an analog of relative utilitarianism for subjective expected utility maximizers.

The proof is modular and yields two intermediary results.
First, we show that restricted monotonicity, faithfulness, no belief imposition, and continuity necessitate linear aggregation with the weight of each agent's belief and utility function being a continuous function of its own preferences.
However, the weights of an agent cannot depend on the other agents' preferences.
Adding independence of redundant acts, the axioms characterize weighted belief-averaging and weighted utilitarian social welfare functions: assign two positive (and possibly different) weights to every agent---one for the belief and one for the utility function---and derive the preferences of any society from the weighted mean of its members' beliefs and the weighted sum of their utility functions.
Lastly, anonymity implies that all agents' weights have to be equal and we get belief-averaging and relative utilitarianism.
\Cref{sec:necessity} discusses the logical independence of the axioms for all three results.
\Cref{sec:literature} gives an overview of related work.
The proofs are in the appendix.

\section{Preferences and Social Welfare Functions}

Let $\Omega$ be a set of \textbf{states} of the world and $\mathcal E$ be a sigma-algebra over $\Omega$.
We refer to elements of $\mathcal E$ as events.
A probability measure $\pi$ on $(\Omega,\mathcal E)$ is \textbf{non-atomic} if for every $E\in\mathcal E$ with $\pi(E) > 0$, there is $F\subset E$, $F\in\mathcal E$, with $0<\pi(F)<\pi(E)$.
A \textbf{belief} is a non-atomic and countably additive probability measure; $\Pi$ denotes the set of all beliefs.
Let $O$ be a set of \textbf{outcomes} endowed with a sigma-algebra.
We assume throughout that $|O| \ge 4$.
An \textbf{act} is a measurable function $f\colon\Omega\rightarrow O$ that maps states to outcomes; $\mathcal A$ is the set of all acts.
A subset of acts $\mathcal A'$ is called \textbf{regular} if $\mathcal A' = \mathcal A(\mathcal E',O') = \{f\in\mathcal A\colon\text{$f$ is $\mathcal E'$-measurable and } f(\Omega)\subset O'\}$ for some sub-sigma-algebra $\mathcal E'\subset\mathcal E$ and subset of outcomes $O'\subset O$.
By $f_*$ we denote the \textbf{push-forward} of $f$ mapping beliefs to probability measures over outcomes.\footnote{For a belief $\pi\in\Pi$ and a measurable set of outcomes $O'\subset O$, $(f_*\pi)(O') = \pi(f^{-1}(O'))$.}
A \textbf{utility function} is a Borel-measurable and bounded function $u\colon O\rightarrow \mathbb R$.
We denote by $\mathcal U$ the set of all utility functions that are normalized to the unit interval, that is, $\inf\{u(x)\colon x\in O\} = 0$ and $\sup\{u(x)\colon x\in O\} = 1$;
$\bar{\mathcal U}$ consists of $\mathcal U$ plus the utility function that is constant equal to $0$.

A \textbf{preference relation} $\s\subset\mathcal A\times\mathcal A$ is a binary relation over acts.
Its strict and symmetric part are $\succ$ and $\sim$, respectively.
The relation with empty strict part is called \textbf{complete indifference}. 
A relation $\s$ maximizes \textbf{subjective expected utility} if for some belief $\pi$ and some utility function $u$,
\[
	f\srel g \text{ if and only if } \int_\Omega (u\circ f) d\pi \ge \int_\Omega (u\circ g) d\pi 
\]
for all acts $f,g$.
In that case, $\pi$ and $u$ represent $\s$.
We denote by $\bar{\mathcal R}$ the set of all preference relations that maximize subjective expected utility; $\mathcal R$ consists of $\bar{\mathcal R}$ minus complete indifference.
For every preference relation in $\mathcal R$, the belief in a representation is unique.
The  utility function is unique up to positive affine transformations.
Hence, for every $\s\in\mathcal R$, there is a unique belief $\pi\in\Pi$ and a unique utility function $u\in\mathcal U$ that represent $\s$.
Complete indifference can be represented by an arbitrary belief $\pi$ and the utility function $u\in\bar{\mathcal U} -\mathcal U$ that is constant equal to 0.
In either case, we write $\ev_{\s}(f) = \int_\Omega (u\circ f) d\pi$ for the expected utility of an act $f$ under $\pi$ and $u$ (which does not depend on the choice of $\pi$ if $\s$ is complete indifference).

Let $I$ be a finite set of \textbf{agents} with $|I| \ge 3$.
Symbols in bold face refer to tuples indexed by $I$.
Every agent has a preference relation $\s[i]\in\bar{\mathcal R}$.
Denote by $\pi_i\in\Pi$ and $u_i\in \bar {\mathcal U}$ a belief and utility function representing $\s[i]$. 
A \textbf{preference profile} $\p\in\bar{\mathcal R}^I$ specifies the preferences of each agent.
We write $\p[-i]$ for the profile obtained by deleting the relation of agent $i$ and $\p[\sim i]$ for the profile obtained by replacing the relation of agent $i$ by complete indifference.
Both notations extend to sets of agents in the obvious way.
An agent is \textbf{concerned} in a profile $\p$ if $\s[i]$ is not complete indifference.
By $I_{\p}$ we denote the set of concerned agents.
A \textbf{social welfare function} (SWF) $\Phi\colon \bar{\mathcal R}^I \rightarrow \bar{\mathcal R}$ maps every preference profile to an element of $\bar{\mathcal R}$.
$\Phi$ is belief-averaging and relative utilitarianism (\baru) if for all $\p\in\bar{\mathcal R}^I$, $\Phi(\p)$ is represented by
\begin{align*}
	\frac1{|I_{\p}|}\sum_{i\in I_{\p}} \pi_i \qquad\text{ and }\qquad \sum_{i\in I_{\p}} u_i\tag{belief-averaging and relative utilitarianism}
\end{align*}
The table in \Cref{sec:referencetable} provides an overview of the notation.

\section{Axioms for Social Welfare Functions}\label{sec:axioms}

We introduce six axioms for SWFs.
The first two, restricted monotonicity and independence of redundant acts, carry the most power in the sense that they rule out other SWFs that have been proposed.

Suppose society is indifferent between two acts $f$ and $g$ for some profile where agent $i$ is completely indifferent.
If $i$ changes its preferences so that its new belief induces the same distribution over outcomes as the belief of society for the original profile, then society should prefer $f$ to $g$ for the new profile if and only if $i$ does.
Formally, for all $i \in I$, $\p\in\bar{\mathcal R}^I$, $\s = \Phi(\p)$, $\s[\sim i] = \Phi(\p[\sim i])$, and $\pi_{\sim i}$ and $\pi_i$ any beliefs representing $\s[\sim i]$ and $\s[i]$,
\begin{align*}
	f \sim_{\sim i} g\text{, } f\srel[i] g\text{, } f_*\pi_{\sim i} = f_*\pi_i\text{, and } g_*\pi_{\sim i} = g_*\pi_i \text{ implies } f \srel g
\end{align*}
Moreover, a strict preference between $f$ and $g$ for agent $i$ in the antecedent implies a strict preference is the consequent.
Note that if $\s[\sim i]$ is complete indifference, then $\s = \s[i]$ since we can choose $\pi_{\sim i} = \pi_i$.

Independence of redundant acts prescribes that if two profiles agree on a set of acts that makes every other act redundant, then the corresponding preferences of society over that set also agree.
We say that a regular set of acts $\mathcal A' = \mathcal A(\mathcal E',O')$ is \textbf{co-redundant} for a profile $\p\in\bar{\mathcal R}^I$ if 
\begin{enumerate}[label=(\roman*)]
	\item \label{item:cored1} every act is unanimously indifferent to some act in $\mathcal A'$; that is, for every $f\in\mathcal A$, there is $g\in\mathcal A'$ such that $f\sim_i g$ for all $i\in I$, and
	\item \label{item:cored2} for all $i \in I$, $\pi_i|_{\mathcal E'}$ is non-atomic on $(\Omega,\mathcal E')$ 
\end{enumerate}
Then, $\Phi$ satisfies independence of redundant acts if for all $\p,\p'\in\bar{\mathcal R}^I$ and every $\mathcal A'\subset \mathcal A$ that is co-redundant for $\p$ and $\p'$, 
\begin{align*}
	\p|_{\mathcal A'} = \p'|_{\mathcal A'} \text{ implies } \Phi(\p)|_{\mathcal A'} = \Phi(\p')|_{\mathcal A'}
	\tag{independence of redundant acts} 
\end{align*}
By part~\ref{item:cored1} of co-redundancy, $\mathcal A'$ and $\mathcal A$ induce the same utility profiles for both $\p$ and $\p'$.
That is, $\{(\ev_{\s[i]}(g))_{i\in I}\colon g\in\mathcal A'\} = \{(\ev_{\s[i]}(f))_{i\in I}\colon f\in\mathcal A\}\subset[0,1]^I$ and similarly for $\p'$.
The point of part~\ref{item:cored2} is that $\s[i]|_{\mathcal A'}$ admits a unique subjective expected utility representation if $i \in I_{\p}$.
\baru would not satisfy independence of redundant acts if this restriction were omitted.\footnote{For beliefs that are not non-atomic, subjective expected utility representations need not be unique up to positive affine transformations of the utility function. If some profile had two different representations, 
$\baru$ may depend on the chosen representation. 
Applying this conclusion to the preferences over the acts $\mathcal A'$, we see that \baru would not satisfy independence of redundant acts if condition \ref{item:cored2} of co-redundancy was omitted.}

Here is a more concise formulation of a strengthening of independence of redundant acts, which explicitly refers to beliefs and normalized utility functions.
For every sub-sigma-algebra $\mathcal E'\subset\mathcal E$, every $O'\subset O$, and any two profiles $\p,\p'\in\bar{\mathcal R}^I$, 
\begin{equation*}
	\makebox[\displaywidth]{$\displaystyle
	 \pi_i|_{\mathcal E'} = \pi_i'|_{\mathcal E'} \text{ is non-atomic on $(\Omega,\mathcal E')$ and } u_i|_{O'} = u_i'|_{O'} \text{ for all $i\in I$ implies } \pi|_{\mathcal E'} = \pi'|_{\mathcal E'} \text{ and } u|_{O'} = u'|_{O'}
	 $}
\end{equation*} 
for some $\pi,\pi'\in\Pi$ and $u,u'\in \bar {\mathcal U}$ representing $\Phi(\p)$ and $\Phi(\p')$.
The antecedent in the definition of independence of redundant acts and part~\ref{item:cored1} of co-redundancy imply that the agents' normalized utilities for outcomes in $O'$ are the same in $\p$ and $\p'$. 
In the stronger formulation, we stipulate this by referring explicitly to normalized utility functions.
This definition is more utilitarian in spirit. 

The remaining four axioms are mostly standard.
Faithfulness requires that society should be completely indifferent whenever all agents are completely indifferent. 
\begin{align*}
	I_{\p} = \emptyset \text{ implies that } \Phi(\p) \text{ is complete indifference}  \tag{faithfulness}
\end{align*}

The axioms above allow for SWFs that ignore the beliefs of some of the agents.
To rule this out, it suffices to assume that no agent can impose its belief on society.
That is, the belief of society is not identical to that of any agent unless the rest of society would arrive at that belief if the agent was completely indifferent.
For all $i \in I$ and $\p\in\bar{\mathcal R}^I$ such that $\Phi(\p[-i])$, $\Phi(\p)$, and $\s[i]$ are not complete indifference, and $\pi_{\sim i}$, $\pi$, $\pi_i$ are the corresponding beliefs,
\begin{align*}
	\pi_{\sim i}\neq \pi_i \text{ implies } \pi\neq\pi_i 	\tag{no belief imposition}
	\label{eq:nondegenerate}
\end{align*}

Continuity requires that small changes in the agents' preferences can only lead to small changes in the preferences of society.
To make this precise, we equip $\bar{\mathcal R}$ with a topology.
The uniform metric $\sup\{|\ev_{\s}(f) - \ev_{\s'}(f)|\colon f\in\mathcal A\}$ induces a topology on $\mathcal R$.
The topology on $\bar{\mathcal R}$ is that of $\mathcal R$ plus the entire set $\bar{\mathcal R}$ (which is thus the only neighborhood of the relation expressing complete indifference).\footnote{This topology on $\bar{\mathcal R}$ in not Hausdorff. In fact, it even violates the T1 separation axiom, which requires that for any two $\s,\s'\in\bar{\mathcal R}$, there are open sets $X,Y$ of $\bar{\mathcal R}$ so that $\s\in X-Y$ and $\s'\in Y-X$. This fails since the only open set containing complete indifference is $\bar{\mathcal R}$. The topology does satisfy the T0 separation axiom, which only requires that some open set contains $\s$ and not $\s'$ \emph{or} $\s'$ and not $\s$.}
Thus, the closure of $\mathcal R$ is $\bar{\mathcal R}$.
The set of profiles $\bar{\mathcal R}^I$ has the product topology of $\bar{\mathcal R}$.
\begin{align*}
	\Phi\text{ is continuous}\tag{continuity}
\end{align*}

Lastly, anonymity prescribes that relabeling the agents does not change society's preferences.
For all $\p\in\bar{\mathcal R}^I$,
\begin{align*}
	\Phi(\p) = \Phi(\p\circ \eta)\text{ for all permutations $\eta$ on $I$}\tag{anonymity}
\end{align*}

\section{Characterization of Belief-Averaging and Relative Utilitarianism}\label{sec:results}

Conceptually, our main result is a characterization of \baru. 
It determines society's preferences for every profile from the average of the concerned agents' beliefs and the sum of their utility functions.
At first sight, this might seem like an incremental extension of \citeauthorhloff{GSS04a}'s linear aggregation result at the cost of much stronger axioms.
This is misleading since their result, like Harsanyi's aggregation theorem, leaves open a choice of weights for \emph{every} profile.
The resulting indeterminacy can go as far as making the result trivially true.
If the agents' utility functions span the space of all utility functions, any utility function can be written as a linear combination of the agents' utility functions.
Then saying that society's utility function is a linear combination of the agents' utility functions imposes no restrictions.
This situation is generic for profiles with at least as many agents as outcomes. 
A theory of preference aggregation that leaves it to the user to choose weights for individuals in every instance of preference aggregation is thus incomplete and sometimes offers no guidance at all.
 
The multi-profile aims at a complete theory.
It allows formulating the idea that preference aggregation should be coherent across different instances by imposing cross-profile axioms.
\Cref{thm:main} shows that with our system of axioms, all agents' weights have to be equal for \emph{all} profiles, and, thus, leaves no free parameters.
The difference between \Cref{thm:main} and \citeauthorhloff{GSS04a}'s result is analogous to the difference between \citeauthorhloff{DhMe99a}'s characterization of relative utilitarianism and \citeauthorhloff{Hars55a}'s social aggregation theorem.

The proof of \Cref{thm:main} proceeds as follows.
First, we only consider the implications of restricted monotonicity, faithfulness, no belief imposition, and continuity.
The former two imply the restricted Pareto condition of \citet{GSS04a} and thus linear aggregation of beliefs and utility functions (with positive weights for utility functions in generic cases by the strict part of restricted monotonicity).
The additional strength of restricted monotonicity lies in the fact that if an agent changes its preferences away from complete indifference, society’s belief for the new profile is an affine combination of the agent’s new belief and that of society for the original profile.
Assuming all beliefs are affinely independent, it follows that the relative weights of the other agents cannot change. 
Thus, their relative weights are independent of the preferences of the agent who changes its preferences.
The analogous statement holds for utility functions.

Now, given any profile, we can apply this conclusion to every agent and the profile obtained by making the agent completely indifferent.
Some algebra shows that the magnitude of the weight for the belief and the utility function of an agent (relative to that of other agents) can only depend on the agent's own preferences.
The signs may however depend on the preferences of the other agents.
For the weights of utility functions any dependence on other agents' preferences vanishes since they have to be positive.
Continuity allows getting the same conclusion for beliefs.
If the weight for an agent's belief ever were to change sign, by continuity, it would have to be zero in some profile (with two concerned agents).
But then the second agent would get to impose its belief, which is ruled out.
We conclude that the weight for an agent's belief cannot change sign.
If it was negative regardless of the agent's preferences, we could find a profile where the belief of society assigns a negative probability to some event, which cannot be. 
In summary, we have that an agent's weight only depends on its own preferences.
This result only applies to profiles with sufficiently non-degenerate utility functions: the utility functions of any pair and any triple of concerned agents are linearly independent.

\begin{restatable}{proposition}{monotonicity}\label{prop:monotonicity}
	Let $\Phi$ be a social welfare function satisfying restricted monotonicity, faithfulness, no belief imposition, and continuity.
	Then, there are continuous functions $\nu_i,\omega_i\colon\mathcal R\rightarrow \mathbb R_{++}$ for $i\in I$ such that for all profiles $\p\in\bar{\mathcal R}^I$ with non-degenerate utilities, $\Phi(\p)$ is represented by
	\begin{align*}
		\frac{1}{\sum_{i\in I_{\p}} \nu_i(\s[i])}\sum_{i\in I_{\p}} \nu_i(\s[i])\pi_i \qquad \text{ and }\qquad \sum_{i\in I_{\p}} \omega_i(\s[i])u_i
	\end{align*}
\end{restatable}

The last step is to show that additionally imposing independence of redundant acts implies that both weights of every agent have to be constant. 
By applying independence of redundant acts to a suitable profile with two concerned agents, one can show that the weights of an agent cannot depend on its belief.
To conclude that they are independent of the utility function as well, we consider a profile in which every act is unanimously indifferent to an act with a range of only three outcomes $\{x_0,x_1,x^*\}$. 
Hence, the set of acts with range $\{x_0,x_1,x^*\}$ is co-redundant.
If the focal agent changes its utility for other outcomes in a way that does not change the image of the profile in utility space, this set remains co-redundant and we can apply independence of redundant acts and conclude that the preferences of society over acts in the co-redundant set do not change.
If the utility functions of the two agents are linearly independent, this implies that the weights of the focal agent remain the same.
Lastly, we construct a path between any pair of utility functions so that neighboring utility functions result in the same weights by the preceding conclusion.

We obtain a characterization of weighted belief-averaging and weighted utilitarian SWFs.
These functions assign two positive weights to every agent, one for their belief and one for their utility function, and determine society's preferences for every profile from the weighted average of the concerned agents' beliefs and the weighted sum of their utility functions.
Crucially, the weights are constant across all profiles, that is, they cannot depend on an agent's own preferences or the preferences of any other agent.
For technical reasons discussed in \Cref{rem:commonutility}, \Cref{prop:main}, as well as the ensuing \Cref{thm:main}, exclude profiles where all agents have equal or completely opposed utility functions, called common utility profiles.

\begin{restatable}{proposition}{main}\label{prop:main}
	Let $\Phi$ be a social welfare function satisfying restricted monotonicity, independence of redundant acts, faithfulness, no belief imposition, and continuity. 
	Then, there are $\bm v,\bm w \in\mathbb R_{++}^I$ such that $\Phi(\p)$ is represented by 
	\begin{align*}
		\frac1{\sum_{I_{\p}} v_i}\sum_{I_{\p}} v_i\pi_i\qquad\text{ and }\qquad \sum_{I_{\p}} w_i u_i
	\end{align*}
	for all non-common utility profiles $\p\in\bar{\mathcal R}^I$; $\bm v$ and $\bm w$ are unique up to multiplication by a positive constant. 
	Conversely, any social welfare function with such a representation for all profiles satisfies the axioms.
\end{restatable}

From here, it is straightforward to conclude that if the SWF is also anonymous, then the weights of all agents have to be equal.
Choosing them to be 1 is without loss of generality and gives \baru.

\begin{restatable}{theorem}{cormain}\label{thm:main}
	Let $\Phi$ be a SWF satisfying restricted monotonicity, independence of redundant acts, faithfulness, no belief imposition, continuity, and anonymity.
	Then, $\Phi$ is \baru on all non-common utility profiles.
\end{restatable}

\begin{remark}[Common utility profiles]\label{rem:commonutility}
	A profile is a common utility profile if there is a utility function $u\in\mathcal U$ such that the utility function of every agent is either $u$ or (equivalent to) $-u$.
	\Cref{thm:main} is false without the restriction to non-common utility profiles.
	 Consider the SWF that coincides with \baru on all non-common utility profiles and yields complete indifference on all common utility profiles.
	 This function satisfies all axioms.\footnote{Restricted monotonicity allows that even if all agents have the same preferences, society's preference is complete indifference or completely reverses the agents' common preferences.}
	 One possibility to avoid the restriction to common utility profiles is a framework with a variable set of agents.
	 \Cref{thm:main} characterizes \baru on every non-common utility profile for an arbitrary set of agents. 
	  Requiring that adding a completely indifferent agent to a profile does not change society's preferences gives the characterization on all profiles.
\end{remark}

\begin{remark}[Identical beliefs]\label{rem:identicalbeliefs}
	Consider the restriction of the framework to profiles where all agents have the same fixed belief $\pi\in\Pi$.
	On this domain, the preferences of each agent over acts are determined by its utility function. 
	Its preferences over acts are thus given by its preferences over the lotteries $f_*\pi$ over outcomes, where $f$ is an act.
	Hence, we get expected utility preferences over lotteries.
	 
	The proofs of our results are still valid on this subdomain (\Cref{rem:identicalbeliefsproof} at the end of the Appendix gives the details).\footnote{If all agents have the same belief, the axioms simplify.
	The restriction on the beliefs in restricted monotonicity becomes vacuous. 
	Independence of redundant acts can be weakened to apply only to regular sets $\mathcal A(\mathcal E,O')$ and Condition~\ref{item:cored2} can be omitted.
	No belief imposition becomes vacuous. 
	Faithfulness, continuity, and anonymity remain the same.}
	Thus, we obtain a characterization of relative utilitarianism similar to that of \citet{DhMe99a}. 
	Restricted monotonicity translates to a condition stronger than their monotonicity axiom.
	Its strict part also implies their non-triviality and no ill will axioms.
	Our continuity is stronger than theirs.
	Independence of redundant acts (independence of redundant alternatives), faithfulness (individualism), and anonymity are the same. 
	An advantage of our proof is that it gives a characterization of weighted utilitarianism (\Cref{prop:identicalbelief}) without anonymity.
\end{remark}

\begin{remark}[Weakening the axioms to profiles with two concerned agents]\label{rem:axiomrestriction}
	\Cref{thm:main} remains valid if we require all axioms except restricted monotonicity (and, of course, faithfulness) to hold only for profiles with two concerned agents.
	Independence of redundant acts, continuity, and anonymity are used only for such profiles.
	No belief imposition is used for other profiles, but it is not hard to check that this could be avoided.
	Anonymity is only used for transpositions.  
\end{remark}

\begin{remark}[Variations of restricted monotonicity]\label{rem:restrictedmonotonicity}
	We discuss two strengthenings of restricted monotonicity.
	Both are satisfied by \baru.
	
	First, one could allow $f \srel[\sim i] g$ in the antecedent. 
	This would remove the restriction to common utility profiles in \Cref{thm:main} since it implies that society's utility function is $u$ in common utility profiles where more agents have utility function $u$ than $-u$.
	
	Second, one could drop the restriction that the focal agent is completely indifferent before changing preferences. 
	More precisely, for any $i\in I$ and any two profiles $\p,\p'\in\bar{\mathcal R}^I$ with $\p[-i] = \p[-i]'$,
	\begin{align*}
		f \sim g, f\sim_i g, f\srel[i]' g, f_*\pi = f_*\pi_i = f_*\pi_i'\text{, and } g_*\pi = g_*\pi_i = g_*\pi_i'\text{ implies } f\srel' g
	\end{align*}
	where $\s = \Phi(\p)$ and $\s' = \Phi(\p')$.	
	This requires that before and after changing preferences, $i$'s belief induces the same distribution over outcomes as society's belief before the change for either of $f$ and $g$.
	One can check that \baru with complete indifference for common utility profiles (see \Cref{rem:commonutility}) satisfies this property. 
	Hence, it does not allow extending \Cref{thm:main} to all profiles.
\end{remark}

\section{Independence of the Axioms}\label{sec:necessity}

We give six examples of SWFs to discuss the independence of the axioms for the results in \Cref{sec:results}.\footnote{\Cref{ex:mon} and \Cref{ex:cont} are adapted from \citet{DhMe99a}.}

\newenvironment{exampleproof}% 
{\begin{addmargin}[0.8cm]{0cm}\small\smallskip}% 
	{\end{addmargin}}
	
\begin{swf}[Weights determined by maximum product of utilities]\label{ex:mon}
		For any profile $\p$, derive society's belief and utility function as follows.
		Society's belief is the average of the beliefs of the concerned agents.
		Consider the closure of the image of the profile in utility space, that is, the closure of $\{(\ev_{\s[i]}(f))_{i\in I_{\p}}\colon f\in\mathcal A\}\subset [0,1]^{I_{\p}}$, and let $(u^i)_{i\in I_{\p}} \in\mathbb R^{I_{\p}}$ be the unique point that maximizes the product of utilities.
		Let society's utility function be the linear combination of the agents' utility functions where the weight of agent $i\in I_{\p}$ is $\Pi_{j\in I_{\p} - \{i\}} u^j$.
		\begin{exampleproof}
			This SWF satisfies independence of redundant acts since the weights for an agent are the same in any two profiles with the same image in utility space.
			It is straightforward to check that faithfulness, no belief imposition, continuity, and anonymity are also satisfied.
			The SWF fails restricted monotonicity since the weight of an agent's utility function can depend on the utility functions of all agents.
			However, the restricted Pareto condition holds since society's belief and utility function are linear combinations of the agents' beliefs and utility functions.
		\end{exampleproof}
\end{swf}

\begin{swf}[Weights determined by distance to fixed preferences]\label{ex:ira}
	Fix some $\s[0]\in\mathcal R$ and let $d(\s[0],\s) = \sup\{\ev_{\s[0]}(f) - \ev_{\s}(f)|\colon f\in\mathcal A\}$ for all $\s$.
	Consider the SWF as in \Cref{prop:monotonicity} with weight functions $\omega_i(\s[i]) = \nu_i(\s[i]) = 2 - d(\s[0],\s[i])$ for every agent $i$.
	\begin{exampleproof}
		The weight functions depend only on an agent's own preferences and are continuous, non-zero, and the same for all agents.
		Hence, this SWF satisfies restricted monotonicity, faithfulness, no-belief imposition, continuity, and anonymity. 
		It fails independence of redundant acts by \Cref{prop:main} since the weights are not constant. 
	\end{exampleproof}
\end{swf}

\begin{swf}[\baru with a phantom agent]\label{ex:faithful}
	Fix some $\s[0]\in\mathcal R$.
	For every profile $\p$, let society's belief be the average of $\pi_0$ and $\pi_i$ for $i\in I_{\p}$; let society's utility function be the sum of $u_0$ and $u_i$ for $i\in I_{\p}$.
	\begin{exampleproof}
		This SWF satisfies restricted monotonicity, independence of redundant acts, no belief imposition, continuity, and anonymity. 
		It violates faithfulness since society's preferences are $\s[0]$ if all agents are completely indifferent.
	\end{exampleproof}
\end{swf}

\begin{swf}[Belief imposition by agent 1]\label{ex:impo}
	Consider the SWF that is defined as \baru except that society's belief is $\pi_1$ whenever agent 1 is concerned.
	\begin{exampleproof}
		 This SWF satisfies restricted monotonicity, independence of redundant acts, faithfulness, and continuity.
		 It fails no belief imposition and anonymity.
	\end{exampleproof}
\end{swf}

\begin{swf}[Special weight for shared preferences]\label{ex:cont}
	 Let $(\alpha_n)$ be a strictly increasing sequence of positive numbers. 
	 Consider the SWF that for a profile $\p$ is represented by $\frac1{\sum_{\s\in\mathcal R} \alpha_{n(\s)}}\sum_{\s\in\mathcal R} \alpha_{n(\s)}\pi_i$ and $\sum_{\s\in\mathcal R} \alpha_{n(\s)}u_i$, where $n(\s)$ is the number of agents in $\p$ with preferences $\s$.
	 \begin{exampleproof}
		This SWF satisfies restricted monotonicity, independence of redundant acts, faithfulness, no belief imposition, and anonymity. 
	 	It is not continuous by \Cref{prop:monotonicity}.
		More explicitly, the weights of an agent are not continuous at the preferences of other agents (unless $\alpha_n = n$ for all $n$).  
	 \end{exampleproof}
\end{swf}

\begin{swf}[Double weight for agent 1]\label{ex:anonymity}
	Consider the SWF that is defined as \baru except that the belief and utility function of agent 1 get weight 2.
	\begin{exampleproof}
		\Cref{prop:main} shows that this SWF satisfies restricted monotonicity, independence of redundant acts, faithfulness, no belief imposition, and continuity. 
		It fails anonymity. 
	\end{exampleproof} 
\end{swf}

Functions~\ref{ex:mon},~\ref{ex:faithful},~\ref{ex:impo}, and~\ref{ex:cont} show that restricted monotonicity, faithfulness, no belief imposition, and continuity cannot be dropped in \Cref{prop:monotonicity} and \Cref{prop:main}.
\Cref{ex:ira} shows that independence of redundant acts is necessary for \Cref{prop:main}.
Functions~\ref{ex:mon},~\ref{ex:ira},~\ref{ex:faithful},~\ref{ex:cont}, and~\ref{ex:anonymity} show that restricted monotonicity, independence of redundant acts, faithfulness, continuity, and anonymity are required for \Cref{thm:main}.
Note that \Cref{ex:impo} is not anonymous.
It is open whether \Cref{thm:main} holds without no belief imposition.  

Our proof requires $|O|\ge 4$ and $|I| \ge 3$ since it relies on profiles with three linearly independent utility functions. 
It is open if the results hold when $|O| = 3$ or $|I| = 2$.

\section{Related Literature on Preference Aggregation Under Uncertainty}\label{sec:literature}

Most closely related to the present paper are three multi-profile results for decision-making under uncertainty. 
\cite{Spru19a} studies preference aggregation when agents are subjective expected utility maximizers but societies need not be.
He characterizes \emph{ex-ante} relative utilitarianism, which ranks acts according to the sum of the agents' expected utilities (when normalizing their utility functions to the unit interval).
That is, according to $\sum_{i\in I}\ev_{\s[i]}(f) = \sum_{i\in I} \int_{\Omega} (u_i\circ f)d\pi_i$.
By contrast, \baru derives a belief of society and ranks acts by the sum of the agents' expected utilities under the belief of society instead of their own belief.
Hence, it yields the ranking induced by $\int_{\Omega}\left(\sum_{i\in I} u_i\circ f\right)d\pi$, where $\pi = \frac1{|I|}\sum_{i\in I}\pi_i$.
The difference is whether the expectation is taken with respect to an agent's own belief or the aggregate belief.
\citeauthorhloff{Spru19a} assumes the full Pareto condition, independence of inessential expansions (a strengthening of independence of redundant acts),\footnote{Independence of inessential expansions requires that if two profiles agree on a set of acts $\mathcal A'$ so that $\mathcal A'$ contains a most-preferred and a least-preferred act for every agent, then the preferences of society over acts in $\mathcal A'$ are the same for both profiles. 
It is stronger than independence of redundant acts since it also applies sets of acts that are not co-redundant.} belief irrelevance (the ranking of constant acts is independent of the beliefs), and that the preferences of society satisfy Savage's sure-thing principle and depend continuously on the agents' preferences.
As a high-level summary, one could say that the assumptions about the SWF are stronger whereas those about the preferences of society are weaker.

\citet{Diet19a} assumes that agents as well as societies are subjective expected utility maximizers as does the present paper (albeit that he works in the framework of \citet{AnAu63a} instead of that of Savage).
He requires that preference aggregation is consistent with Bayesian updating.
That is, preference aggregation and Bayesian updating commute.
This asserts that the preferences of society do not depend on whether new information arrives before or after the aggregation.
\citeauthorhloff{Diet19a} shows that consistency with Bayesian updating, continuity, and a weakening of the restricted Pareto condition imply geometric-linear aggregation: society's belief is a \emph{geometric} mean of the agents' beliefs and society's utility function is a linear combination of the agents' utility functions.
The weight of an agent in either of these combinations can depend on the profile of utility functions, but not on the beliefs.

Early work on collective decision-making under uncertainty by \citet{HyZe79b} studies functions that output a collective \emph{choice} instead of a preference relation of society.
Again, all preferences are subjective expected utility-maximizing.
They show that no social choice function satisfies the following properties: acts in the choice set are expected utility-maximizing for a belief and utility function of society, this belief (utility function) of society depends only on the agents' beliefs (utility functions), Pareto-dominated acts are never chosen, and no agent can dictate the belief.

Much of the literature on single-profile results is inspired by \citet{Mong95a} and \citet{GSS04a} and focuses on the Pareto condition.
Recall that \citeauthorhloff{Mong95a} proved that the Pareto condition can hold only if all agents have the same belief.
Several authors have examined to what extent \citeauthorhloff{Mong95a}'s impossibility is robust to deviations from Savage's subjective expected utility framework.
\citet{Mong98a} showed that it persists in \citeauthorhloff{AnAu63a}'s model of subjective expected utility as well as if utilities may be state-dependent as long as the preferences identify a unique belief.
More radical deviations, such as replacing Savage's sure thing principle by state-wise dominance, have been considered by \citet{BDM04a,GMV08a,Fleu09a,ChHa14a,MoPi15a}.
The general verdict from these works is that the impossibility pops up even under weaker assumptions about individual and collective preferences.  

\citet{GSS14a} consider no-betting-Pareto dominance, which is in between the full and the restricted Pareto condition.
It states that a society should prefer one act to another if Pareto dominance in the usual sense holds and there is a belief such that if all agents held that belief, they would also unanimously prefer the first act to the second.
They argue that no-betting-Pareto dominance characterizes situations in which agents can benefit from trade.

\citet{AlGa16a} assume that societies have max-min expected utility preferences \citep[cf.][]{GiSc89a}, where acts are compared based on their minimal expected utility within a set of beliefs.
They show that the restricted Pareto condition and a Pareto condition on the beliefs imply that the utility functions are aggregated linearly and the set of beliefs of society are convex combinations of the agents' beliefs.

\section{Discussion}

We conclude the paper with some remarks.

\paragraph{Spurious unanimities and complementary ignorance}
The motivation for restricting the monotonicity condition to acts about which beliefs agree is, as for the restricted Pareto condition, to avoid spurious unanimities.
\citet{MoPi19a} demonstrate that if the agents have different private information, even these restricted conditions can be vulnerable to a sort of spurious unanimities they call \emph{complementary ignorance}.
An example similar to theirs illustrates the issue.

Consider two agents who can choose between two bets on the winner of a race between three horses. 
The first bet pays \$1 to each agent if horse 1 or horse 2 wins and \$0 otherwise.
The second bet pays \$1 to each agent if horse 3 wins and \$0 otherwise.
Assume further that both agents have a uniform prior belief over the winner.
Now the first agent privately learns that horse 1 cannot race and updates their belief accordingly so that the agent assigns probability one-half each to horse 2 and horse 3 winning. 
This agent is then indifferent between both bets.
Likewise, the second agent privately learns that horse 2 cannot race, updates their belief, and is thus also indifferent between both bets.
Note that either bet induces the same distribution over outcomes (payoff \$1 or \$0) for both agents' updated beliefs.
Hence, the restricted monotonicity (or Pareto) condition applies and stipulates that, collectively, the agents should be indifferent between both bets.
However, given both agents' private information, it is certain that horse 3 will win, thus making the second bet preferable.

We conclude that if the agents have different private information, even the restricted monotonicity condition may still be too strong.
Geometric aggregation of beliefs as studied by \citet{Diet19a} is not susceptible to complementary ignorance (in the example, the geometric mean of the agents' beliefs assigns probability 1 to horse 3 winning). 
Note, however, that if the agents can exchange their private information and have a common interest in doing so (as in our example), the problem of complementary ignorance vanishes. 
Moreover, heterogeneity in beliefs may not be due to differences in information. 
For example, when not assuming common knowledge of all agents' rationality, an agent may not change their belief even if the agent knew the other agents' beliefs.
The well-known example of two gentlemen contemplating a duel due to \citet{GSS04a} is plausibly such a case.

\paragraph{Equal range of utilities}
Relative utilitarianism normalizes the agents' utility functions so that the difference between the highest and lowest utility outcome is the same across all agents. 
For this to be a tenable conclusion, the set of feasible outcomes needs to contain a best-possible and a wort-possible outcome for every agent. 
Then we can plausibly postulate that every agent experiences the same difference in satisfaction between their most-preferred and least-preferred outcome.
\citet{DhMe99a} eloquently sum up this issue when stating
``[$\dots$] to apply RU meaningfully, one has to, and it suffices to, consider a set of alternatives sufficiently encompassing as to include, besides the actual alternatives of interest, each person's best and worst alternative within the ``universal'' set $A$, limited only by feasibility and justice. This leads in turn to a concept of ``absolute utility'': the ``correct'' scaling of an individual's utility is determined solely by his own preferences and by the philosophy of the state adopted.''

One cause for this specific normalization is anonymity.
If we do not assume assume anonymity, society's utility function can be a \emph{weighted} sum of the agents $(0,1)$-normalized utility functions (cf. \Cref{prop:main}).
This is equivalent to normalizing each agent's utility function to the interval $(0,w_i)$ (where $w_i$ is the weight of the agent) and taking the unweighted sum of the utility functions instead.
It can account for situations in which outcomes have far-reaching consequences for some agents and only minor consequences for other agents. 
But then one has to decide about the weights given to the agents, or, equivalently, the range of their experienced satisfaction, which is a non-trivial ethical challenge.

What remains even without assuming anonymity is that every agent gets the same weight across all decision situations. 
We have discussed in the introduction and when sketching the proof of \Cref{prop:main} that this is a consequence of independence of redundant acts.
So long as the set of outcomes is fixed (as we assume), this conclusion seems reasonable, perhaps even desirable.
If, however, one considers situations involving different sets of outcomes, it loses its appeal.
And so does independence of redundant acts since it conflicts with the requirement that the collective preferences do not change when restricting the set of outcomes.

\paragraph{Interpretation of beliefs}
	Two of the axioms are not formulated purely on the level of preferences: restricted monotonicity and no belief imposition reference the belief in the expected utility representation of a preference relation.
	While it is possible to state these axioms in terms of preferences only, such a formulation would be clumsy and obscure the intuitions.
	So for these axioms to be meaningful, it is not enough that the preferences can be \emph{represented} by expected utility maximization (that is, satisfy Savage's postulates~1 through~7).
	The representing beliefs and utility functions also need to be the decision maker's \emph{rationale} for the preference relation.

	One prerequisite for the latter is that the utility functions are state-independent as assumed in Savage's framework.
	Whether state-independent utilities are a tenable assumption depends on the context.
	State-independence requires that the decision-maker is only affected by the outcomes and not by the state.
	This is true in good approximation for, say, investment decisions where the decision-maker is only interested in the realized asset value and not the causes for the outcome.
	By contrast, if the decision-maker cares about the causes for an outcome, as, for example, when the decision-maker's health depends on the state, state-independence may fail.
	We refer to \citet{DrRu99a} for a critical discussion of state-(in)dependent utilities.
	If utilities were allowed to be state-dependent, the belief in the expected utility representation would not be unique, thus making it ambiguous how to interpret an axiom that references a single belief. 
	To regain the uniqueness of beliefs, one needs additional structural assumptions on the preferences \citep[see, e.g.,][]{KSV83a}.

	A more general critique is that beliefs and utility functions are nothing but mathematical objects used to concisely represent preferences.
	This is a general critique of subjective expected utility maximization and not specific to this paper. 
	\citet[][Section 1.3]{GSS14a} and \citet[][Section 5.1]{GGSS14a} address the subject is some detail and we will not engage in it further.

\paragraph{Properties of belief-averaging and relative utilitarianism}
One can ask which axiomatic properties \baru satisfies beyond those assumed in \Cref{thm:main} and which ones it violates.
As mentioned, restricted monotonicity and faithfulness together imply the restricted Pareto condition of \citet{GSS04a}. 
It is also clear from the impossibility result of \citet{Mong97a} that \baru does not satisfy the full Pareto condition (except for profiles where all agents have the same belief).

It is less obvious that the axioms imply that the beliefs and the utility functions are aggregated separately. 
That is, society's belief does not depend on the agents' utility functions and likewise for society's utility function.
The SWFs characterized in \Cref{prop:monotonicity} do not, in general, have this property. 
So the separability relies on the independence of redundant acts axiom.
Similarly, \baru is neutral in the sense that it does not distinguish between the outcomes. 
A priori, it could not be ruled out that, say, a larger weight is given to agents who rank acts that yield certain outcomes at the top of their preferences.
Independence of redundant acts again plays a key role in establishing neutrality.

\paragraph{Alternative interpretations of the model}
We have interpreted the formal framework as modeling collective decisions under uncertainty. 
An alternative interpretation is the aggregation of inter-temporal preferences.
Suppose there is a continuous time axis and at each point in time, some outcome is realized.
Moreover, there are multiple agents, each of which has preferences over the outcome streams ranking them by their discounted utility for some discounting function and some utility function on the outcomes.
The agents' preferences over the outcome streams are then to be aggregated into a collective preference.
For example, the members of a household may try to decide on a household consumption plan over some (possibly indefinite) time horizon.
\Citet{JaYa14a} discuss further examples and employ a similar model.

This situation fits our framework as follows.
Each point in time corresponds to a state.
Hence, outcome streams are identified with acts and discounting functions with beliefs.
\begin{align*}
	\text{states} \quad&\longleftrightarrow\quad \text{points in time}\\
	\text{acts} \quad&\longleftrightarrow\quad \text{outcome streams}\\
	\text{beliefs} \quad&\longleftrightarrow\quad \text{discounting functions}
\end{align*}
We point out that the analogue of averaging beliefs is averaging \emph{normalized} discounting functions.
Since beliefs are probability measures, the corresponding normalization is to scale discounting functions so that their integral over time is 1 (so discounting functions have to be in $L^1$).\footnote{This is also the normalization we would have derived for beliefs had we allowed them to be finite positive measures instead of probability measures.}

There is a vast literature on aggregation of inter-temporal preferences and no attempt is made to cover it here.
Whether our axioms should hold and \baru is suited for the aggregation task is a separate issue.
For example, there is an inherent ordering on points in time, which the current framework with an abstract state space does not capture.
The purpose of this discussion to merely to highlight that the presented framework is general enough to model various applications.

\section*{Acknowledgments}
This material is based on work supported by the Deutsche Forschungsgemeinschaft under grant {BR~5969/1-1}.
The author thanks Jean Baccelli, Tilman B\"orgers, Francesc Dilm\'e, David Easley, Loren Fryxell, Johannes H\"orner, Ian Jewitt, Michel Le Breton, Wolfgang Pesendorfer, Dominik Peters, Marcus Pivato, Evgenii Safonov, Yves Sprumont, and Omer Tamuz, an anonymous associate editor, and two anonymous referees for helpful comments.
A previous version of this paper has been presented in the Public Economics and Microeconomic Theory Seminar at Cornell University, the Microeconomic Theory Seminar at Princeton University, the D-TEA Workshop 2021, and the Conference on New Directions in Social Choice in St. Petersburg.

\appendix
\section*{APPENDIX: Proofs}

An SWF $\Phi$ gives rise to two functions $\phi$ and $\psi$, which take a profile $\p\in\bar{\mathcal R}^I$ to a belief $\phi(\p)\in\Pi$ and a utility function $\psi(\p)\in\bar{\mathcal U}$ of society. 
Whenever society is completely indifferent, $\phi(\p)$ can be arbitrary. 

Preferences are invariant under multiplication of the belief by a positive constant and positive affine transformations of utility functions.
For measures $\pi,\pi'$ on $(\Omega,\mathcal E)$ and utility functions $u,u'$ on $O$, we write $\pi\equiv \pi'$ if $\pi = \alpha\pi'$ for some $\alpha > 0$ and $u\equiv u'$ if $u = \alpha u' + \beta$ for $\alpha > 0$ and $\beta\in\mathbb R$.

The \textbf{belief dimension} of $\p$ is the dimension of $\{(\pi_i(E))_{i\in I_{\p}}\colon E \in \mathcal E\} \subset\mathbb R^{I_{\p}}$ and the \textbf{utility dimension} is the dimension of $\{(u_i(p))_{i\in I_{\p}}\colon p \text{ a probability measure on $O$}\}$.
Hence, the belief (utility) dimension is the maximal number of affinely (linearly) independent beliefs (utility functions).
We say that a profile $\p\in\bar{\mathcal R}^I$ has \textbf{non-degenerate beliefs} if the beliefs of any pair and any triple of agents in $I_{\p}$ are affinely independent.
We define \textbf{non-degenerate utilities} similarly.
A profile is \textbf{non-degenerate} if it has non-degenerate beliefs and utilities.
The assumption $|O| \ge 4$ ensures that non-degenerate profiles exist and are dense in $\bar{\mathcal R}^I$.

\Cref{prop:main} requires that we find $\bm v,\bm w\in R^I_{++}$ such that $\phi(\p) \equiv \sum_{i\in I_{\p}} v_i\pi_i$ and $\psi(\p) \equiv \sum_{i\in I_{\p}} w_iu_i$ for all profiles $\p\in\bar{\mathcal R}^I$ with non-degenerate beliefs.
The proof proceeds in three steps.
First, we examine the implications of restricted monotonicity in conjunction with faithfulness and no belief imposition.
These axioms imply that on all non-degenerate profiles, the weights assigned to an agent's belief and utility function can only depend on the agent's own preferences. 
The weights for beliefs may be negative however. 
Second, we add continuity, which allows us to rule out negative weights and to extend the obtained representation to all profiles with non-degenerate utilities (\Cref{prop:monotonicity}).
Lastly, independence of redundant acts implies that the weights of an agent cannot depend on the agent's preferences either and thus have to be constant across all profiles.
Deriving \Cref{thm:main} from \Cref{prop:main} is easy.

\section{Implications of Restricted Monotonicity}\label{sec:monotonicity}

The proofs that the weight of the belief and the weight of the utility function of an agent can only depend on the agent's preferences in Sections~\ref{sec:beliefaggregation} and~\ref{sec:tasteaggregation} proceed along the same lines.
For the most part, the proof for beliefs requires more work, since we cannot rule out negative weights.
Thus, we advise readers interested in the proofs to take a look at \Cref{sec:tasteaggregation} first.

\subsection{Aggregation of Beliefs}\label{sec:beliefaggregation}

The first lemma states that if an agent changes its preferences away from complete indifference, the society's belief for the obtained profile is an affine combination of society's belief for the original profile and the new belief of the agent.
The restricted Pareto condition of \citet{GSS04a} already implies that society's belief is an affine combination of its members' beliefs.
The additional strength of restricted monotonicity lies in the fact that no matter what the new belief of the agent is, it is always combined with the same belief of society for the original profile.
Since we assume that an agent cannot impose its belief on the society, the agent's weight in the affine combination cannot be 1.

\begin{lemma}\label{lem:beliefaffine1}
	Assume that $\Phi$ satisfies restricted monotonicity and no belief imposition.
	Let $i\in I$ and $\p\in\bar{\mathcal R}^I$.
	Then, $\phi(\p) = (1-\alpha)\phi(\p[\sim i]) + \alpha\pi_i$ for some $\alpha\in\mathbb R$.
	If $\Phi(\p[\sim i])$ is not complete indifference, then $\alpha\neq 1$.
\end{lemma}

\begin{proof}
	Let $\s = \Phi(\p)$ and $\s[\sim i] = \Phi(\p[\sim i])$.
	Restricted monotonicity implies that $f\sim g$ whenever $f\sim_{\sim i} g$, $f\sim_i g$, $f_*\phi(\p[\sim i]) = f_*\pi_i$, and $g_*\phi(\p[\sim i]) = g_*\pi_i$.
	Thus, (the two-agent case of) Theorem 1 of \citet[][]{GSS04a} implies that $\phi(\p) = (1-\alpha)\phi(\p[\sim i]) + \alpha\pi_i$ for some $\alpha\in\mathbb R$.	
	If $\pi_i = \phi(\p[\sim i])$ or $\Phi(\p)$ is complete indifference, we can choose $\alpha$ arbitrarily.
	If $\s[i]$ is complete indifference, we have $\alpha = 0$.
	Otherwise, if $\Phi(\p[\sim i])$ is not complete indifference, $\phi(\p) \neq \pi_i$, since $\Phi$ satisfies no belief imposition, and so $\alpha\neq 1$.
\end{proof}

\begin{lemma}\label{lem:beliefaffine2}
	Assume that $\Phi$ satisfies restricted monotonicity, faithfulness, and no belief imposition.
	Let $\p\in\bar{\mathcal R}^I$ where $\Phi(\p)$ is not complete indifference.
	Then $\phi(\p) = \sum_{i\in I_{\p}} v_i\pi_i$ for some $\bm v\in\mathbb R^{I_{\p}}$ with $\sum_{i\in I_{\p}} v_i = 1$.
	Moreover, if $(\pi_i)_{i\in I_{\p}}$ are affinely independent and $(u_i)_{i\in I_{\p}}$ are linearly independent, then $\bm v\in(\mathbb R-\{0\})^{I_{\p}}$ and $\bm v$ is unique.
\end{lemma}

\begin{proof}
	Since $\Phi$ is faithful, we have that $\Phi(\p[\sim I_{\p}])$ is complete indifference.
	Letting one agent after another in $I_{\p}$ change its preferences from complete indifference to $\s[i]$ and applying \Cref{lem:beliefaffine1} at each step, we get $\phi(\p) = \sum_{i\in I_{\p}}  v_i\pi_i$ for some $\bm v\in\mathbb R^{I_{\p}}$ with $\sum_{i\in I_{\p}} v_i = 1$.
	
	If $(\pi_i)_{i\in I_{\p}}$ are affinely independent, $\bm v$ is unique.
	We prove by induction over $|I_{\p}|$ that $ v_i\neq 0$ for all $i$.
	If $|I_{\p}| = 1$, then $v_i = 1$ for $i \in I_{\p}$ is forced.
	Now suppose that $|I_{\p}| > 1$ and let $i, j\in I_{\p}$.
	By the induction hypothesis, we have $\phi(\p[\sim i]) = \sum_{k\in I_{\p}-\{i\}} v'_k\pi_k$ and $\phi(\p[\sim j]) = \sum_{k\in I_{\p}-\{j\}} v''_k\pi_k$ for some $\bm v'\in(\mathbb R-\{0\})^{I_{\p} - \{i\}}$ and $\bm v''\in(\mathbb R-\{0\})^{I_{\p} - \{j\}}$.
	The assumption that $(u_i)_{i\in I_{\p}}$ are linearly independent ensures that $\Phi(\p[\sim i])$ and $\Phi(\p[\sim j])$ are not complete indifference.
	\Cref{lem:beliefaffine1} implies that 
	\[
		\phi(\p) = (1-\alpha)\phi(\p[\sim i]) + \alpha\pi_i = (1-\beta)\phi(\p[\sim j]) + \beta\pi_j
	\]
	for some $\alpha,\beta\in\mathbb R-\{1\}$.
	We set $\bm v = ((1-\alpha)\bm v',\alpha) = ((1-\beta)\bm v'',\beta)$, where $\alpha$ and $\beta$ appear in position $i$ and $j$ respectively.
	Since $\alpha\neq 1$ and $ v'_k\neq 0$, it follows that $ v_k\neq 0$ for all $k\in I_{\p} - \{i\}$. 
	Similarly, $\beta\neq 1$ implies that $ v_k\neq 0$ for all $k\in I_{\p} - \{j\}$. 
\end{proof}

For later use, we prove a fact for profiles with belief dimension at least 3.

\begin{lemma}\label{lem:belieftriple}
	Assume that $\Phi$ satisfies restricted monotonicity and no belief imposition.
	Let $\p\in\bar{\mathcal R}^I$ where $\Phi(\p)$ is not complete indifference.
	If $\p$ has belief dimension at least 3, there are distinct $i, j\in I$ such that $\phi(\p[\sim i, j]), \pi_i$, and $\pi_j$ are affinely independent.
\end{lemma}

\begin{proof}
		Since $\p$ has belief dimension at least 3, we may assume that $\pi_1,\pi_2$, and $\pi_3$ are affinely independent.
		So $\phi(\p)$ cannot be in the affine hull of all three pairs from $\{\pi_1,\pi_2,\pi_3\}$, for if say $\phi(\p)$ is in the affine hull of $\{\pi_1,\pi_2\}$ and $\{\pi_1,\pi_3\}$, then $\phi(\p) = \pi_1$ and so is not in the affine hull of $\{\pi_2,\pi_3\}$.
		Assume that $\phi(\p)$ is not in the affine hull of $\{\pi_1,\pi_2\}$.
		Then \Cref{lem:beliefaffine1} implies that $\phi(\p[\sim 1,2])$ is not in the affine hull of $\{\pi_1,\pi_2\}$.
		Since $\pi_1\neq\pi_2$, $\phi(\p[\sim 1,2]), \pi_1$, and $\pi_2$ are affinely independent.
\end{proof}

\Cref{lem:beliefaffine2} ensures that society's belief is an affine combination of the agents' beliefs. 
To show that $\phi$ has the form claimed in \Cref{prop:monotonicity}, we have to prove that the relative weight of an agent in this affine combination depends only on the agent's own belief and utility function.
For now, we only derive a weaker conclusion, which allows negative weights.

\begin{lemma}\label{lem:beliefindependence}
	Assume that $\Phi$ satisfies restricted monotonicity, faithfulness, and no belief imposition.
	Then there are for all $i \in I$ and $J \subset I$, $\nu_i\colon \mathcal R\rightarrow\mathbb R - \{0\}$ and $\sigma_i^J\colon\mathcal R^J\rightarrow\mathbb \{-1,1\}$ such that for all non-degenerate $\p\in\bar{\mathcal R}^I$, $\phi(\p) \equiv \sum_{i\in I_{\p}} \sigma^{I_{\p}}_i(\p)\nu_i(\s[i])\pi_i$.
	Moreover, for all $i,j\in I_{\p}$, and non-degenerate $\p\in\bar{\mathcal R}^I$,
	\begin{align*}
		\frac{\sigma^{I_{\p}}_i(\p)}{\sigma^{I_{\p}}_j(\p)} = \frac{\sigma^{\{i,j\}}_i(\p[\sim (I - \{i,j\})])}{\sigma^{i,j}_j(\p[\sim (I - i,j)])}
	\end{align*} 
\end{lemma}

\begin{proof}
	We assume throughout the proof that all considered profiles are non-degenerate.
	For $l\in I-\{1,2\}$, let $I_l = \{1,2,l\}$ and $\bar{\mathcal R}_l\subset\bar{\mathcal R}^I$ be the set of profiles with $I_{\p} = I_l$.
	Let $l\in I-\{1,2\}$ be arbitrary and fix some $\tilde\p\in\bar{\mathcal R}_l$.
	By \Cref{lem:beliefaffine2}, there are unique functions $\kappa_i\colon\bar{\mathcal R}_l\rightarrow\mathbb R-\{\bm 0\}$ for all $i \in I_l$ such that $\phi(\p) = \sum_{i\in I_l} \kappa_i(\p)\pi_i$ for all $\p\in\bar{\mathcal R}_l$. 
	For $i,j\in I_l$ and $\p\in\bar{\mathcal R}_l$, let
	\begin{align*}
		\lambda^{i,j}(\p) = \frac{\kappa_j(\p[-j], \tilde{\s[j]}) \kappa_i(\p)}{\kappa_j(\tilde\p) \kappa_i(\p[-j], \tilde{\s[j]})} \quad\text{ and }\quad \nu_i(\s[i]) = \frac{|\kappa_i(\p)|}{|\lambda^{i,j}(\p)|}  \quad\text{ and }\quad \sigma^{I_l}_i(\p) = \sgn(\kappa_i(\p))
	\end{align*}
	The fact that the $\kappa_i$ map to $\mathbb R-\{\bm 0\}$ ensures that $\lambda^{i,j}$ is well-defined.
	We show that $\nu_i$ is independent of $j$ and $\p[-i]$ and thus well-defined.
	We proceed in three steps.
	Note that the projection of $\bar{\mathcal R}_l$ to $\mathcal R$ that returns the preferences of $i$ is onto, and so $\nu_i$ is a function on all of $\mathcal R$.
	
	\begin{step}\label{step:ratio}
		Let $k\in I_l-\{i, j\}$.
		We show that $\frac{\kappa_i(\p)}{\kappa_j(\p)}$ is independent of $\s[k]$.
		To this end, let $\p'\in\bar{\mathcal R}_l$ such that $\p[-k]' = \p[-k]$.
		By \Cref{lem:beliefaffine1}, we have that 
		\begin{align*}
			\phi(\p) &= (1-\alpha)\phi(\p[\sim k]) + \alpha\pi_k = \kappa_i(\p)\pi_i  + \kappa_j(\p)\pi_j + \kappa_k(\p)\pi_k\text{, and}\\
			\phi(\p') &= (1-\beta)\phi(\p[\sim k]') + \beta\pi_k' = \kappa_i(\p')\pi_i  + \kappa_j(\p')\pi_j + \kappa_k(\p')\pi_k',
		\end{align*}
		for some $\alpha,\beta\in\mathbb R - \{1\}$.
		Affine independence of $\pi_1,\pi_2,\pi_l$ and $\pi_1,\pi_2,\pi'_l$ implies that $\kappa_i(\p)\pi_i + \kappa_j(\p)\pi_j \equiv\phi(\p[\sim k]) = \phi(\p[\sim k]')\equiv \kappa_i(\p')\pi_i + \kappa_j(\p')\pi_j$.
		In particular, $\frac{\kappa_i(\p)}{\kappa_j(\p)} = \frac{\kappa_i(\p')}{\kappa_j(\p')}$.
	\end{step}
	
	\begin{step}\label{step:lambda}
		We show that $\lambda^{i,j}(\p)$ is independent of $i$ and $j$.
		This is tedious, but only uses \Cref{step:ratio}.
		Let $k\in I_l-\{i, j\}$.
		First we show independence of $j$.
		\begin{align*}
			\lambda^{i,j}(\p) &= \frac{\kappa_j(\p[-j], \tilde{\s[j]}) \kappa_i(\p)}{\kappa_j(\tilde\p) \kappa_i(\p[-j], \tilde{\s[j]})}\\
			&= \frac{\kappa_j(\p[-j,k],\tilde{\s[k]},\tilde{\s[j]}) \kappa_i(\p)}{\kappa_j(\bm{\tilde\pi}) \kappa_i(\p[-j,k],\tilde{\s[k]},\tilde{\s[j]})}\\
			&= \frac{\kappa_j(\p[-j,k],\tilde{\s[k]},\tilde{\s[j]}) \kappa_i(\p) \kappa_k(\p[-j,k],\tilde{\s[k]},\tilde{\s[j]})}{\kappa_j(\tilde\p) \kappa_i(\p[-j,k],\tilde{\s[k]},\tilde{\s[j]})\kappa_k(\p[-j,k],\tilde{\s[k]},\tilde{\s[j]})}\\
			&= \frac{\kappa_j(\tilde\p) \kappa_k(\p[-k],\tilde{\s[k]}) \kappa_i(\p)}{\kappa_j(\tilde\p) \kappa_i(\p[-k],\tilde{\s[k]})\kappa_k(\tilde\p)} = \lambda^{i,k}(\p)
		\end{align*}
Verifying independence of $i$ is very similar.
		\begin{align*}
			\lambda^{i,j}(\p) &= \frac{\kappa_j(\p[-j], \tilde{\s[j]}) \kappa_i(\p)}{\kappa_j(\tilde\p) \kappa_i(\p[-j], \tilde{\s[j]})}\\
			&= \frac{\kappa_j(\tilde{\p[-k,i]},\s[k],\s[i]) \kappa_i(\p)}{\kappa_j(\tilde\p) \kappa_i(\tilde{\p[-k,i]},\s[k],\s[i])}\\
			&= \frac{\kappa_j(\tilde{\p[-k,i]},\s[k],\s[i]) \kappa_i(\p)\kappa_k(\tilde{\p[-k,i]},\s[k],\s[i])}{\kappa_j(\tilde\p) \kappa_i(\tilde{\p[-k,i]},\s[k],\s[i])\kappa_k(\tilde{\p[-k,i]},\s[k],\s[i])}\\
			&= \frac{\kappa_j(\p[-j], \tilde{\s[j]}) \kappa_i(\p)\kappa_k(\p)} {\kappa_j(\tilde\p) \kappa_i(\p)\kappa_k(\p[-j], \tilde{\s[j]})}
			 = \lambda^{k,j}(\p)
		\end{align*}
	\end{step}
	\begin{step}
		We show that $\nu_i(\p)$ is independent of $\p[-i]$ and $j$.
		\begin{align*}
			\nu_i(\s[i]) = \frac{|\kappa_i(\p)|}{|\lambda^{i,j}(\p)|}
			= \frac{|\kappa_j(\tilde\p) \kappa_i(\p[-j], \tilde{\s[j]})|}{|\kappa_j(\p[-j], \tilde{\s[j]})|} = 
			\frac{|\kappa_j(\tilde\p) \kappa_i(\tilde{\p[-i]}, \s[i])|}{|\kappa_j(\tilde{\p[-i]}, \s[i])|},
		\end{align*}
		where we use \Cref{step:ratio} for the last equality.
		The last term is independent of $\p[-i]$ and, by \Cref{step:lambda}, the second term is independent of $j$.
	\end{step}
	Now it is easy to see that 
	\begin{align*}
		\phi(\p) = \sum_{i\in I_l}\kappa_i(\p)\pi_i \equiv \sum_{i\in I_l}\frac{\kappa_i(\p)}{|\lambda^{i,j_i}(\p)|}\pi_i =\sum_{i\in I_l} \sigma_i^{I_l}(\p)\nu_i(\s[i])\pi_i,
	\end{align*}
	where $j_i\in I_l - \{i\}$ for all $i$.
	For the second equality, we used the fact that $\lambda^{i,j}$ is independent of $i$ and $j$.
	
	Since $l$ was arbitrary, we have now defined $\nu_i$ for each $i\in I$.
	However, we have defined $\nu_1$ and $\nu_2$ multiple times, once for each $l\in I-\{1,2\}$.
	So we have to check that these definitions are not conflicting. 
	It follows from \Cref{lem:beliefaffine1} that the ratio between $\nu_1$ and $\nu_2$ is the same for each triple $\{1,2,l\}$.
	Thus, we can define $\nu_1$ and $\nu_2$ as obtained for, say, $l=3$ and scale the triples $(\nu_1,\nu_2,\nu_l)$ obtained for the remaining $l$ appropriately.
	
\begin{step}\label{step:sigmadef}
	Now we define the $\sigma_i$ for the remaining profiles.
	Our strategy will be to first define it for profiles with two concerned agents, and then inductively for all non-degenerate profiles.
	At each point, we maintain that 
	\begin{align*}
		\frac{\sigma^{I_{\p}}_i(\p)}{\sigma^{I_{\p}}_j(\p)} = \frac{\sigma^{\{i,j\}}_i(\p[\sim (I - i,j)])}{\sigma^{\{i,j\}}_j(\p[\sim (I - i,j)])}
	\end{align*}
	which we will refer to as the \emph{ratio condition} on the $\sigma_i$.
	We omit the superscript in expressions like $\sigma_i^{I_{\p}}(\p)$ from now on, since it is clear from the profile.
	
	Let $\p\in\bar{\mathcal R}^I$. 
	If $|I_{\p}| = 2$, say $I_{\p} = \{1,2\}$, then $\phi(\p) = \alpha_1\pi_1 + \alpha_2\pi_2$ for unique $\alpha_1,\alpha_2\in\mathbb R-\{0\}$.
	We define $\sigma_i(\p) = \sgn(\alpha_i)$ for $i\in I_{\p}$.
	
	Now assume that $|I_{\p}|\ge 3$ and that we have defined the $\sigma_i$ such that the ratio condition holds on all profiles with fewer than concerned $|I_{\p}|$ agents.
	Let $i\in I_{\p}$ such that $\pi_i\neq\phi(\p[\sim i])$, which exists by \Cref{lem:belieftriple}.
	We show that there is $s\in\{-1,1\}$ such that $\frac{s}{\sigma_j(\p[\sim i])} = \frac{\sigma_i(\p[\sim I - i,j])}{\sigma_j(\p[\sim I - i,j])}$ for all $j\in I_{\p}-\{i\}$.
	If not, there are $j,k$ which require $s=1$ and $s=-1$ respectively.
	It is not hard to see that then there must be $j,k$ with this property such that $\pi_i$, $\pi_j$, and $\pi_k$ are affinely independent.
	Then we have
	\begin{equation*}
		\makebox[\displaywidth]{$\displaystyle
		\frac{\sigma_j(\p[\sim I - j,k])}{\sigma_k(\p[\sim I - j,k])} = \frac{\sigma_j(\p[\sim i])}{\sigma_k(\p[\sim i])} = -\frac{\sigma_j(\p[\sim I - i,j])}{\sigma_i(\p[\sim I - i,j])} \frac{\sigma_i(\p[\sim I - i,k])}{\sigma_k(\p[\sim I - i,k])}
	=-\frac{\sigma_j(\p[\sim I - i,j,k])}{\sigma_i(\p[\sim I - i,j,k])} \frac{\sigma_i(\p[\sim I - i,j,k])}{\sigma_k(\p[\sim I - i,j,k])} = -\frac{\sigma_j(\p[\sim I - i,j,k])}{\sigma_k(\p[\sim I - i,j,k])}
	$}
	\end{equation*}
	The first equality uses that $\sigma_j^{I_{\p} - \{i\}}, \sigma_k^{I_{\p} - \{i\}}$ satisfy the ratio condition, the second the choice of $j$ and $k$, and the third that $\sigma_i^{\{i,j,k\}},\sigma_j^{\{i,j,k\}},\sigma_k^{\{i,j,k\}}$ satisfy the ratio condition. 
	But this contradicts that $\sigma_i^{\{i,j,k\}},\sigma_j^{\{i,j,k\}},\sigma_k^{\{i,j,k\}}$ satisfy the ratio condition.
	Thus, we can find $s$ as required.

	By \Cref{lem:beliefaffine1} and the choice of $i$, we have that $\phi(\p) = (1-\alpha)\phi(\p[\sim i]) + \alpha\pi_i$ for some unique $\alpha\in\mathbb R-\{1\}$.
	If $\alpha < 1$, we set $\sigma_{i'}^I = \sigma_{i'}^{I-\{i\}}$ for $i'\in I - \{i\}$ and $\sigma_i^I = s$; if $\alpha > 1$, set $\sigma_{i'}^I = -\sigma_{i'}^{I-\{i\}}$ for $i'\in I - \{i\}$ and $\sigma_i^I = -s$.
\end{step}

	We still have to make sure that these definitions of the $\nu_i$ and $\sigma_i$ are consistent with $\phi$.
	For $\p\in\bar{\mathcal R}^I$, let $\bar\phi(\p) \equiv \sum_{i\in I_{\p}} \sigma_i^{I_{\p}}(\p)\nu_i(\s[i])\pi_i$.
	We show by induction over $|I_{\p}|$, starting at $|I_{\p}| = 3$, that $\phi$ and $\bar\phi$ agree on all profiles $\p$.

	We make two observations.
\begin{step}\label{step:beliefsubprofile}
	Let $\p\in\bar{\mathcal R}^I$ with belief dimension 3 such that $\bm\pi_{\sim i}$ has belief dimension 2.
	If $\phi$ and $\bar\phi$ agree on $\p$, then they also agree on $\p[\sim i]$.
	By \Cref{lem:beliefaffine1} and the assumption, we have
	\begin{align*}
		\phi(\p) = (1-\alpha)\phi(\p[\sim i]) + \alpha\pi_i \equiv \sum_{j\in I_{\p}-\{i\}}\sigma_j(\p)\nu_j(\s[j])\pi_j + \sigma_i(\p)\nu_i(\s[i])\pi_i
	\end{align*}
	for some $\alpha\in\mathbb R-\{1\}$.
	Since $\pi_i$ is not in the affine hull of $(\pi_j)_{j\in I_{\p}-\{i\}}$, we have to have $(1-\alpha)\phi(\p[\sim i])\equiv \sum_{j\in I_{\p}-\{i\}}\sigma_j(\p)\nu_j(\s[j])\pi_j$.
	If $\alpha < 1$, then by definition, $\sigma_j(\p[\sim i]) = \sigma_j(\p)$ for all $j\in I_{\p}-\{i\}$, and so $\phi(\p[\sim i]) \equiv \sum_{j\in I_{\p}-\{i\}}\sigma_j(\p[\sim i])\nu_j(\s[j])\pi_j\equiv\bar\phi(\p[\sim i])$.
	If $\alpha > 1$, then $\sigma_j(\p[\sim i]) = -\sigma_j(\p)$ for all $j$, and again $\phi(\p[\sim i]) = \bar\phi(\p[\sim i])$ follows. 
\end{step}

\begin{step}\label{step:beliefextend}
	Let $\p\in\bar{\mathcal R}^I$ and $i,j\in I_{\p}$.
	If $\phi$ and $\bar\phi$ agree on $\p[\sim i]$ and $\p[\sim j]$ and $\bar\phi(\p[\sim i,j]), \pi_i$, and $\pi_j$ are affinely independent, then they also agree on $\p$.
	\Cref{lem:beliefaffine1} implies that
	\begin{align*}
		\phi(\p) = (1-\alpha)\phi(\p[\sim j]) + \alpha \pi_j = (1-\beta)\phi(\p[\sim i]) + \beta \pi_i
	\end{align*}
	for some $\alpha,\beta\in\mathbb R-\{1\}$.
	To make notation less cumbersome, we write $\sigma_k(\p[\sim J]) = \sigma_k^J$ and $\nu_k(\s[k]) = \nu_k$ for $k\in I_{\p}$ and $J\subset I_{\p} - \{i\}$ for the rest of this step.
	Four cases arise, depending on whether $\alpha$ and $\beta$ are greater of smaller than 1.
	\begin{case}
		Assume $\alpha,\beta < 1$.
		By definition, we have that $\sigma_k = \sigma_k^j$ for all $k\in I_{\p}-\{j\}$ and $\sigma_k = \sigma_k^i$ for all $k\in I_{\p}-\{i\}$.
		In particular, $\sigma_k^i = \sigma_k^j$ for $k\in I_{\p} - \{i,j\}$.
		Moreover, either $\sigma_k^{ij} = \sigma_k^i$ for all $k\in I_{\p} - \{i,j\}$ or $\sigma_k^{ij} = -\sigma_k^i$ for all $k$.
		Let $s = 1$ in the former case and $s=-1$ otherwise.
		Then, 
		\begin{align*}
			\phi(\p)\equiv s\underbrace{\left(\sum_{k\in I_{\p}-\{i,j\}}\sigma_k^{ij}\nu_k\pi_k\right)}_{\phi(\p[-i,j])} + \sigma_i^j\nu_i\pi_i + \alpha'\pi_j
			= s\left(\sum_{k\in I_{\p}-\{i,j\}}\sigma_k^{ij}\nu_k\pi_k\right) + \beta'\pi_i + \sigma_j^i\nu_j\pi_j
		\end{align*}
		for some $\alpha',\beta'\in\mathbb R$.
		Affine independence implies that $\alpha' = \sigma_j^i\nu_j = \sigma_j\nu_j$.
		Moreover, $\sigma_i^j = \sigma_i$ and $s\sigma_k^{ij} = \sigma_k^i = \sigma_k$ for $k \in I_{\p} - \{i,j\}$.
		So $\phi(\p) = \sum_{k\in I_{\p}} \sigma_k\nu_k\pi_k$, which concludes this case.
	\end{case}
	\begin{case}
		Assume $\alpha > 1$ and $\beta < 1$.
		By definition, we have that $\sigma_k = -\sigma_k^j$ for all $k\in I_{\p}-\{j\}$ and $\sigma_k = \sigma_k^i$ for all $k\in I_{\p}-\{i\}$.
		In particular, $\sigma_k^i = -\sigma_k^j$ for $k\in I_{\p} - \{i,j\}$.
		Moreover, either $\sigma_k^{ij} = \sigma_k^j$ for all $k\in I_{\p} - \{i,j\}$ or $\sigma_k^{ij} = -\sigma_k^j$ for all $k$.
		Let $s = 1$ in the former case and $s=-1$ otherwise.
		Then, 
		\begin{align*}
			\phi(\p)&\equiv -\phi(\p[\sim j]) + \alpha' \pi_j \equiv -s\sum_{k\in I_{\p}-\{i,j\}}\sigma_k^{ij}\nu_k\pi_k- \sigma_i^j\nu_i\pi_i + \alpha'\pi_j\text{, and}\\
			&\equiv \phi(\p[\sim i]) + \beta' \pi_i \equiv -s\sum_{k\in I_{\p}-\{i,j\}}\sigma_k^{ij}\nu_k\pi_k + \beta'\pi_i + \sigma_j^i\nu_j\pi_j
		\end{align*}
		for some $\alpha',\beta'\in\mathbb R$.
		The second equality in the second line follows from $-s\sigma_k^{ij} = -\sigma_k^j = \sigma_k^i$ for $k\in I_{\p}-\{i,j\}$.
		Affine independence implies that $\alpha' = \sigma_j^i\nu_j = \sigma_j\nu_j$.
		Moreover, $-\sigma_i^j = \sigma_i$ and $-s\sigma_k^{ij} = -\sigma_k^j = \sigma_k$.
		So $\phi(\p) = \sum_{k\in I_{\p}} \sigma_k\nu_k\pi_k$.
	\end{case}
	The remaining two cases are analogous to the two we have examined and therefore omitted.
\end{step}

\begin{step}\label{step:nondegenerate}
We have shown that $\phi$ and $\bar\phi$ agree on profiles with $I_{\p} = \{1,2,l\}$ for all $l$, which we use for the base case $|I_{\p}| = 3$.
Now let $\p\in\bar{\mathcal R}^I$ with $I_{\p} = \{1,i, j\}$ for distinct $i,j\in I-\{1\}$.
Observe that $\phi$ and $\bar\phi$ agree on the profiles $\p[\sim i]$ and $\p[\sim j]$ of $\p$ by \Cref{step:beliefsubprofile}.
Since $\bar\phi(\p[\sim i,j]) = \pi_1, \pi_i$, and $\pi_j$ are affinely independent, \Cref{step:beliefextend} implies that $\phi$ and $\bar\phi$ agree on $\p$.
With a second application of the same argument, we get that $\phi$ and $\bar\phi$ agree on all profiles with $|I_{\p}|=3$.
The preceding argument also takes care of profiles with $|I_{\p}| = 2$.
Profiles with $|I_{\p}| = 1$ are covered by \Cref{lem:beliefaffine2}.

Now we deal with the case $|I_{\p}|\ge 4$.
Since $\p$ is non-degenerate, $\p[\sim i]$ has belief dimension at least 3 for all $i\in I_{\p}$.
	The induction hypothesis implies that $\phi$ and $\bar\phi$ agree on $\p[\sim i]$ and $\p[\sim j]$ for all $i,j\in I_{\p}$.
	The argument in the proof of \Cref{lem:belieftriple} also applies to $\bar\phi$, and so we can choose distinct $i, j\in I_{\p}$ such that $\bar\phi(\p[\sim i, j]), \pi_j$, and $\pi_i$ are affinely independent.
	\Cref{step:beliefextend} implies that $\phi$ and $\bar\phi$ agree on $\p$.
\end{step}	
\end{proof}

\subsection{Aggregation of Utility Functions}\label{sec:tasteaggregation}

It is useful to first clarify the linear algebra on $\bar{\mathcal U}$.
Elements of $\bar{\mathcal U}$ are normalized representatives of a class of utility functions, consisting of all its positive affine transformations.
Thus, we say that $(u_i)_{i\in I}$ are linearly independent if their span does not include any utility function that is equivalent to the $0$ element of $\bar{\mathcal U}$, that is, any constant utility function.

We show that for any profile and any agent, society's utility function is a linear combination of the agent's utility function and that of the society for the profile with the agent completely indifferent.
If the latter two utility functions are not equal or negations of each other, the agent has positive weight in this linear combination.
In \Cref{lem:utilitylinear2}, we leverage this fact to prove that society's utility function is a \emph{positive} linear combination of the agents' utility functions.

\begin{lemma}\label{lem:utilitylinear1}
	Assume that $\Phi$ satisfies restricted monotonicity.
	Let $\p\in\bar{\mathcal R}^I$ and $i \in I_{\p}$.
	Then $\psi(\p) \equiv \alpha\psi(\p[\sim i]) + \beta u_i$ for some $\alpha,\beta\in\mathbb R$.
	Moreover, if $u_i\not\equiv\pm \psi(\p[\sim i])\not\equiv 0$, then $\beta > 0$ and $\alpha,\beta$ are unique up to a common positive factor.
\end{lemma}
 
\begin{proof}
	Let $\s = \Phi(\p)$, $\s[\sim i] = \Phi(\p[\sim i])$, and $\pi_{\sim i} = \phi(\p[\sim i])$.
	We start in the same way as for \Cref{lem:beliefaffine1}.
	Restricted monotonicity implies that $f\sim g$ whenever $f\sim_{\sim i} g$, $f\sim_i g$, $f_*\pi_{\sim i} = f_*\pi_i$, and $g_*\pi_{\sim i} = g_*\pi_i$.
	Thus, it follows from Theorem 1 of \citet[][]{GSS04a} that $\psi(\p) \equiv \alpha\psi(\p[\sim i]) + \beta u_i$ for some $\alpha,\beta\in\mathbb R$.
	
	If $u_i\not\equiv\pm \psi(\p[\sim i]) \not\equiv 0$, then $\alpha,\beta$ are unique up to a common positive factor.
	Moreover, we can find probability distributions $p$ and $q$ on $O$ with finite support such that $\psi(\p[\sim i])(p) = \psi(\p[\sim i])(q)$ and $u_i(p) > u_i(q)$.
	A theorem of \citet{Liap40a} allows us to construct acts $f$ and $g$ that induce the distributions $p$ and $q$ under $\pi_{\sim i}$ and $\pi_i$, respectively.
	That is, $p = f_*\pi_{\sim i} = f_*\pi_i$ and $q = g_*\pi_{\sim i} = g_*\pi_i$.
	Thus, $f\sim_{\sim i} g$ and $f\succ_ig$.
	The strict part of restricted monotonicity then implies $f\succ g$.
	From \Cref{lem:beliefaffine1}, we know that $\phi(\p)$ is an affine combination of $\pi_{\sim i}$ and $\pi_i$, and so $f_*\phi(\p) = p$ and $g_*\phi(\p) = q$.
	It follows that $f\succ g$ if and only if $\psi(\p)(p) > \psi(\p)(q)$.
	Thus, $\beta > 0$.
\end{proof}

\begin{lemma}\label{lem:utilitylinear2}
	Assume that $\Phi$ satisfies restricted monotonicity and faithfulness.
	Then for every $\p\in\bar{\mathcal R}^I$, $\psi(\p) \equiv \sum_{i\in I_{\p}}w_iu_i$ for some $\bm w\in\mathbb R^{I_{\p}}$.
	If $(u_i)_{i\in I_{\p}}$ are linearly independent, $\bm w\in\mathbb R^{I_{\p}}_{++}$ and $\bm w$ is unique up to a positive factor.
\end{lemma}

\begin{proof}
	The first part is a straightforward corollary of \Cref{lem:utilitylinear1}. 
	For the second part, assume that $(u_i)_{i\in I_{\p}}$ are linearly independent.
	Let $\bm w\in\mathbb R^{I_{\p}}$ such that $\psi(\p) \equiv \sum_{i\in I_{\p}}  w_iu_i$.
	Linear independence implies that $\bm w$ is unique up to a positive factor and $u_i\not\equiv\pm\psi(\p[\sim i]) \not\equiv 0$ for all $i\in I_{\p}$.
	Thus, \Cref{lem:utilitylinear1} implies that $\psi(\p)\equiv \alpha\psi(\p[\sim i]) + \beta u_i$ for $\alpha\in\mathbb R$ and $\beta > 0$.
	Since $\psi(\p[\sim i])$ is a linear combination of $(u_j)_{j\in I_{\p} - \{i\}}$ and $\bm w$ is unique up to a positive factor, it follows that $w_i > 0$.
\end{proof}

The next lemma is the analogue of \Cref{lem:belieftriple}.
Its proof is similar and therefore omitted.

\begin{lemma}\label{lem:utilitytriple}
	Assume that $\Phi$ satisfies restricted monotonicity.
	Let $\p\in\bar{\mathcal R}^I$ such that $\Phi(\p)$ is not complete indifference.
	If $\p$ has utility dimension at least 3, there are distinct $i, j\in I_{\p}$ such that $\psi(\p[\sim i, j]), u_i$, and $u_j$ are linearly independent.
\end{lemma}

In general, the $w_i$ in \Cref{lem:utilitylinear2} may depend on $\p$.
The content of the next lemma is that $ w_i$ must not depend on $\p[\sim i]$ for non-degenerate profiles.

\begin{lemma}\label{lem:utilityindependence}
	Assume that $\Phi$ satisfies restricted monotonicity and faithfulness.
	Then there are $\omega_i\colon \mathcal R\rightarrow \mathbb R_{++}$ for $i\in I$ such that $\psi(\p) \equiv \sum_{i\in I_{\p}} \omega_i(\s[i])u_i$ for all $\p\in\bar{\mathcal R}^I$ with non-degenerate utilities.
\end{lemma}

\begin{proof}
	\setcounter{step}{0}
	\setcounter{case}{0}
	We assume throughout the proof that all considered profiles have non-degenerate utilities.
	The first part is very similar to the construction of the $\nu_i$ in the proof of \Cref{lem:beliefindependence}.
	For $l\in I -\{1,2\}$, let $I_l = \{1,2,l\}$ and $\bar{\mathcal R}_l$ be the set of all profiles $\p\in\bar{\mathcal R}^I$ such that $I_{\p} = I_l$.
	Let $l\in I-\{1,2\}$ be arbitrary and fix some $\tilde{\p}\in\bar{\mathcal R}_l$.
	By \Cref{lem:utilitylinear2}, there are functions $\omega_i\colon\bar{\mathcal R}_l\rightarrow\mathbb R_{++}$ such that $\psi(\p) \equiv \sum_{i\in I_l} \omega_i(\p)u_i$ for all $\p\in\bar{\mathcal R}_l$ and the $\omega_i(\p)$ is unique up to a common positive factor.
	For $i,j\in I_l$ and $\p\in\bar{\mathcal R}_l$, let 
	\begin{align*}
		\lambda^{i,j}(\p) = \frac{\omega_j(\p[-j],\tilde{\s[j]}) \omega_i(\p)}{\omega_j(\tilde\p) \omega_i(\p[-j],\tilde{\s[j]})}\quad\text{ and }\quad	\omega_i(\s[i]) = \frac{\omega_i(\p)}{\lambda^{i, j}(\p)}
	\end{align*}
	The fact that the $\omega_i$ and $\omega_j$ map to $\mathbb R_{++}$ ensures that $\lambda^{i,j}$ is well-defined and positive.
	Note that the projection of $\bar{\mathcal R}_l$ to $\mathcal R$ that returns the preferences of $i$ is onto, and so $\omega_i$ is a function on all of $\mathcal R$.
	With the same arguments as in the proof of \Cref{lem:utilityindependence}, we can show that $\frac{\omega_i(\p)}{\omega_j(\p)}$ is independent of $\s[k]$ for $k\in I_l-\{i,j\}$, that $\lambda^{i,j}$ is independent of $i$ and $j$, and that $\omega_i$ is well-defined.
	Then we have
	\begin{align*}
		\psi(\p) \equiv \sum_{i\in I_l} \omega_i(\p) u_i \equiv \sum_{i\in I_l} \frac{\omega_i(\p)}{\lambda^{i,j_i}(\p)} u_i = \sum_{i\in I_l} \omega_i(\s[i]) u_i,
	\end{align*}
	where $j_i\in I_l - \{i\}$ for all $i$.
	
	Since $l$ was arbitrary, we have now defined $\omega_i$ for each $i\in I$.
	However, we have defined $\omega_1$ and $\omega_2$ multiple times, once for each $l\in I-\{1,2\}$.
	So we have to check that these definitions are not conflicting. 
	It follows from \Cref{lem:beliefaffine1} that the ratio between $\omega_1$ and $\omega_2$ is the same for each triple $\{1,2,l\}$.
	Thus, we can define $\omega_1$ and $\omega_2$ as obtained for, say, $l=3$ and scale the triples $(\omega_1,\omega_2,\omega_l)$ obtained for the remaining $l$ appropriately.
	
	The $\omega_i$ define a function that returns a utility function of society for every profile.
	For $\p\in\bar{\mathcal R}^I$, let $\bar\psi(\p) \equiv \sum_{i\in I_{\p}} \omega_i(\s[i])u_i$.
	We make two observations.
	
\begin{step}\label{step:utilitysubprofile}
	Let $\p$ be a profile of utility dimension 3 such that $\p[\sim i]$ has utility dimension 2.
	If $\psi$ and $\bar\psi$ agree on $\p$ and $\bar\psi(\p[\sim i]) \not\equiv 0$, then they also agree on $\p[\sim i]$.
	By \Cref{lem:utilitylinear1} and the assumption, we have
	\begin{align*}
		\psi(\p) \equiv \alpha\psi(\p[\sim i]) + \beta u_i \equiv \sum_{j\in I_{\p}-\{i\}}\omega_j(\s[j])u_j +  \omega_i(\s[i])u_i
	\end{align*}
	for some $\alpha,\beta\in\mathbb R$.
	Note that $\bar\psi(\p[\sim i]) \equiv \sum_{j\in I_{\p}-\{i\}}\omega_j(\s[j])u_j\not\equiv 0$ implies that $\alpha \neq 0$.
	Since $u_i$ is not in the span of $(u_j)_{j\in I_{\p}-\{i\}}$, we then get $\psi(\p[\sim i]) \equiv \bar\psi(\p[\sim i])$.
\end{step}

\begin{step}\label{step:utilityextend}
	Let $\p$ be a profile and $i,j\in I_{\p}$.
	If $\bar\psi(\p[\sim i,j]), u_i$, and $u_j$ are linearly independent and $\psi$ and $\bar\psi$ agree on $\p[\sim i]$ and $\p[\sim j]$, then they also agree on $\p$.
	
	\Cref{lem:utilitylinear1} and \Cref{lem:utilitylinear2} imply that
	\begin{align*}
		\psi(\p) &\equiv \alpha\psi(\p[\sim j]) + \beta u_j \equiv \overbrace{\alpha\sum_{k\in I_{\p}-\{i,j\}} \omega_k(\s[k])u_k + \alpha\omega_i(\s[i])u_i}^{\equiv\bar\psi(\p[\sim j])} + \gamma u_j\text{, and}\\
		\psi(\p) &\equiv \alpha\psi(\p[\sim i]) + \beta' u_i \equiv \underbrace{\alpha\sum_{k\in I_{\p}-\{i,j\}} \omega_k(\s[k])u_k}_{\equiv\bar\psi(\p[\sim i,j])} + \gamma' u_i + \alpha\omega_j(\s[j])u_j
	\end{align*}
	for some $\alpha,\beta,\beta',\gamma,\gamma'\in\mathbb R_{++}$.
	Linear independence of $\bar\psi(\p[\sim i,j]), u_i$, and $u_j$ implies that $\gamma = \alpha\omega_j(\s[j])$, and so $\psi$ and $\bar\psi$ agree on $\p$.
\end{step}
	
	Now we can finish the proof.
	We show by induction over $|I_{\p}|$, starting at $|I_{\p}| = 3$, that $\psi$ and $\bar\psi$ agree on all profiles with non-degenerate utilities.
	
	The base case is $|I_{\p}| = 3$.
	Let $\p\in\bar{\mathcal R}^I$.
	First assume that $I_{\p} = \{1,i,j\}$ for distinct $i,j\in I-\{1\}$.
	We have shown that $\psi$ and $\bar\psi$ agree on profiles $\p'$ with $I_{\p'} = \{1,2,l\}$ for any $l$.
	Thus \Cref{step:utilitysubprofile} implies that they agree on $\p[\sim i]$ and $\p[\sim j]$.
	Moreover, $\bar\psi(\p[\sim i,j]) = u_1, u_i$, and $u_j$ are linearly independent.
	So \Cref{step:utilityextend} implies that $\psi$ and $\bar\psi$ agree on $\p$.
	A second iteration of the same argument implies that they agree on any profile with three concerned agents.
	The preceding argument also takes care of the case $|I_{\p}| = 2$.
	Profiles with $|I_{\p}| = 1$ are covered by \Cref{lem:utilitylinear2}.
	
Now we deal with the case $|I_{\p}|\ge 4$.
\begin{case}
	Suppose $\psi(\p) = 0$.
	Assume for contradiction that $\bar\psi(\p) \not\equiv 0$.
	Let $i,j\in I_{\p}$.
	\Cref{lem:utilitylinear1} implies that $\psi(\p[\sim i])\equiv\pm u_i$.
	The induction hypothesis implies $\psi(\p[\sim i]) \equiv \bar\psi(\p[\sim i])$.
	Hence, $\bar\psi(\p)\equiv \bar\psi(\p[\sim i]) + \alpha u_i \equiv\pm u_i$.
	Similarly, we get $\bar\psi(\p)\equiv \pm u_j$.
	This contradicts $u_i\not\equiv\pm u_j$, which holds since $\p$ has non-degenerate utilities.
\end{case}

\begin{case}\label{case:utilitynondegenerate}
	Suppose $\psi(\p) \not= 0$.
	By \Cref{lem:utilitytriple} we can choose distinct $i, j\in I$ such that $\psi(\p[\sim i, j]), u_i$, and $u_j$ are linearly independent.
	The induction hypothesis implies that $\psi$ and $\bar\psi$ agree on $\p[\sim i,j]$, $\p[\sim i]$, and $\p[\sim j]$.	
	So \Cref{step:utilityextend} implies that $\psi$ and $\bar\psi$ agree on $\p$.
\end{case}
\end{proof}

\section{Implications of Continuity}\label{sec:continuity}

Recall that the topology on $\mathcal R$ is induced by the uniform metric $\sup\{|\ev_{\s}(f) - \ev_{\s'}(f)|\colon f \in\mathcal A\}$.
The topology on $\bar{\mathcal R}$ consists of the open sets in the topology on $\mathcal R$ plus the set $\bar{\mathcal R}$.
We define topologies on $\Pi$ and $\bar{\mathcal U}$ analogously.
For $\pi,\pi'\in\Pi$, the uniform metric $\sup\{|\pi(E) - \pi'(E)|\colon E\in\mathcal E\}$ gives a topology on $\pi$.
For $u,u'\in\mathcal U$ (note the absence of constant utility function 0), we also use the uniform metric $\sup\{|u(x) - u'(x)|\colon x\in O\}$.
The topology on $\bar{\mathcal U}$ is that of $\mathcal U$ plus the entire set $\bar{\mathcal U}$. 
So the only neighborhood of the utility function that is constantly 0 is the set $\bar{\mathcal U}$ itself.
This is the topology $\bar{\mathcal U}$ inherits from the space of all utility functions equipped with the uniform metric when taking the quotient by positive affine transformations.

These topologies make the mappings from preference relations to beliefs and utility functions continuous.
Likewise, the inverse operation, mapping a pair of belief and utility function to a preference relation is continuous.\footnote{Lemma 5 of \citet{Diet19a} is the equivalent of this statement in the framework of \citet{AnAu63a}.}
To ease the notation in the proof of the next lemma, when $E$ is an event and $x,y$ are outcomes, we write $xEy$ for the act which yields $x$ for states in $E$ and $y$ for states in $\Omega - E$.

\begin{lemma}\label{lem:cont}
	The correspondence $\pi(\s)$ and the function $u(\s)$ mapping $\s\in\bar{\mathcal R}$ to the beliefs and the utility function representing $\s$ are (upper-hemi) continuous.
	Moreover, the function $\s(\pi,u)$ mapping each pair of belief and utility function to the preference relation it induces is continuous.
\end{lemma}

\begin{proof}
	Let $(\s^n)$ be a sequence converging to $\s$ in $\bar{\mathcal R}$.
	For each $n$, let $\pi^n\in\pi(\s^n)$ and $u^n = u(\s^n)$.
	
	First we show that $(u^n)$ converges to $u = u(\s)$.
	Let $x\in O$ and $f_x$ be the constant act returning $x$ for all states.
	We have $\sup\{|u^n(x)-u(x)|\colon x\in O\} = \sup\{|\ev_{\s^n}(f_x)-\ev_{\s}(f_x)|\colon x\in O\}$, and so $(u^n)$ converges uniformly to $u$.
	
	Second, we need to show that if $(\pi^n)$ converges to $\pi'\in\Pi$, then $\pi'\in\pi(\s)$.
	If $u = 0$, then $\pi(\s) = \Pi$ and there is nothing to show.
	Otherwise, $\pi(\s) = \{\pi\}$ for some $\pi\in\Pi$. 
	We show that $(\pi^n)$ converges to $\pi$, which implies $\pi = \pi'$.
	Since $u\neq 0$, we can choose $x,y\in O$ such that $u(x) > u(y)$.
	 
	Then, for large enough $n$,
	\begin{align*}
		\sup\left\{\left|\pi^n(E)-\pi(E)\right|\colon E\in\mathcal E\right\} &= \sup\left\{\left|\frac{\ev_{\s^n}(xEy)-u^n(y)}{u^n(x) - u^n(y)} - \frac{\ev_{\s}(xEy)-u(y)}{u(x) - u(y)}\right|\colon E\in \mathcal E\right\} \\
		&\le \frac2{u(x) - u(y)}\sup\left\{\left|\ev_{\s^n}(xEy) - \ev_{\s}(xEy)\right|\colon E\in \mathcal E\right\}
	\end{align*}
	and so $(\pi^n)$ converges uniformly to $\pi$.
	
	Conversely, assume that $(\pi^n)$ and $(u^n)$ converge to $\pi$ and $u$, respectively. 
	Let $\s^n = \s(\pi^n,u^n)$ and $\s = \s(\pi,u)$ be the induced preference relations.
	If $u = 0$, then $\s$ is complete indifference.
	Since $\bar{\mathcal R}$ is the only neighborhood of complete indifference in $\bar{\mathcal R}$, $(\s^n)$ trivially converges to $\s$.
	Suppose $u \neq 0$.
	For $\epsilon > 0$, let $n_0\in\mathbb N$ such that for $n\ge n_0$, $\sup\{|\pi^n(E) - \pi(E)|\colon E\in\mathcal E\}<\frac{\epsilon}2$ and $\sup\{|u^n(x) - u(x)|\colon x\in O\}<\frac{\epsilon}2$.
	($n_0$ exists since $u \neq 0$.)
	Then, for all $n\ge n_0$ and $f\in\mathcal A$,
	\begin{align*}
		|\ev_{\s^n}(f)-\ev_{\s}(f)| &= \left|\int_\Omega (u^n\circ f)d\pi^n -\int_\Omega (u\circ f)d\pi\right|\\
		&\le \left|\int_\Omega (u^n\circ f - u\circ f)d\pi^n\right| + \left|\int_\Omega (u\circ f)d\pi^n-\int_\Omega (u\circ f)d\pi\right|\\
		&< \frac{\epsilon}2 + \frac{\epsilon}2 = \epsilon
	\end{align*}
	Thus, $\ev_{\s^n}$ converges uniformly to $\ev_{\s}$.
\end{proof}

The basis for \Cref{prop:monotonicity} are \Cref{lem:beliefindependence} and \Cref{lem:utilityindependence}.
\Cref{lem:beliefindependence} allows negative weights for the beliefs and only applies to non-degenerate profiles.
Assuming that the SWF is continuous eliminates the first issue and partially the second.

\monotonicity*

\begin{proof}
	\setcounter{case}{0}
	\setcounter{step}{0}
	For all $i\in I$ and $J\subset I$, let $\nu_i$, $\omega_i$, $\sigma_i^J$ be the functions obtained from \Cref{lem:beliefindependence} and \Cref{lem:utilityindependence}.
	\begin{step}\label{step:mucont}
		We show that each $\omega_i$ is continuous.
		Let $i,j\in I$ and $\p\in\bar{\mathcal R}^I$ be non-degenerate with $I_{\p} = \{i,j\}$. 
		Let $(\s[i]^n)$ be a sequence in $\mathcal R$ converging to $\s[i]\in\mathcal R$ and $\p^n = (\p[-i],\s[i]^n)$.
		By assumption, $\Phi$ is continuous, and by \Cref{lem:cont}, the mapping from preference relations to the corresponding utility functions is continuous.
		Thus, by \Cref{lem:utilityindependence}, $u^n = \psi(\p^n) \equiv \omega_i(u_i^n)u_i^n + \omega_j(u_j)u_j$ converges to $u = \psi(\p) \equiv \omega_i(u_i)u_i + \omega_j(u_j)u_j$.
		First, $\omega_i(\s[i]^n)$ is bounded, as otherwise, a subsequence of $(u^n)$ would converge to $u_i$.
		But this is impossible, since $\omega_j(u_j)\neq 0$ and $u_i\not\equiv\pm u_j$.
		Now if $\alpha$ is an accumulation point of $(\omega_i(u^n_i))$, then $\alpha u_i + \omega_j(u_j) u_j\equiv u$, since $(u^n)$ converges to $u$.
		But $\alpha u_i + \omega_j(u_j) u_j \equiv \beta u_i + \omega_j(u_j) u_j$ if and only if $\alpha = \beta$.
		So $(\omega_i(u_i^n))$ is bounded and has a unique accumulation point.
		Thus, it converges to $\omega_i(u_i)$.
	\end{step}
	
	\begin{step}\label{step:lambdacont}
		Let $i,j\in I$ and $\p\in\bar{\mathcal R}^I$ be non-degenerate with $I_{\p} = \{i,j\}$. 
		We show that $\sigma_i^{\{i,j\}}\nu_i$ is continuous at $\p$.
		(For convenience, we will omit the superscript ${\{i,j\}}$ from now on.)
		Let $(\p^n)$ be a sequence of profiles with $I_{\p^n} = \{i,j\}$ converging to $\p$.
		Let $\alpha^n = \sigma_i(\p^n)\nu_i(\s[i]^n)$ and $\beta^n = \sigma_j(\p^n)\nu_j(\s[j]^n)$, and $\alpha = \sigma_i(\p)\nu_i(\s[i])$ and $\beta = \sigma_j(\p)\nu_j(\s[j])$.
		
		First we prove convergence when $j$'s preferences remain constant at $\s[j]$.
		Let $\tilde{\p}^n = (\p[-i],\s[i]^n)$, $\tilde\alpha^n = \sigma_i(\tilde{\p}^n)\nu_i(\s[i]^n)$, and $\tilde\beta^n = \sigma_j(\tilde{\p}^n)\nu_j(\s[j])$.
		We need to show that $(\tilde\alpha^n)$ converges to $\alpha$.
		Note that $\tilde\beta^n$ can only vary in sign but not in absolute value.
		Since $\Phi$ and the correspondence mapping preference relations to the corresponding beliefs are continuous, we have that $\tilde\alpha^n\pi_i^n + \tilde\beta^n\pi_j\equiv \phi(\tilde\p^n)\rightarrow \phi(\p) \equiv\alpha\pi_i + \beta\pi_j$.
		With the same reasoning as in \Cref{step:mucont}, we get that $(\tilde\alpha^n)$ is bounded and has a unique accumulation point.
		Thus, it converges to $\alpha$.
		Similarly, $\sigma_j(\p[-j],\s[j]^n)\nu_j(\s[j]^n)$ converges to $\beta$.
		
		Now we show that $(\alpha^n)$ converges $\alpha$.
		We already know that the sequences of absolute values of $(\alpha^n)$ and $(\beta^n)$ converge to $\alpha$ and $\beta$, respectively. 
		So any subsequence $(\alpha^{n_k},\beta^{n_k})$ such that all $\alpha^{n_k}$ and all $\beta^{n_k}$ have the same sign converges.
		By the same reasoning as in the previous paragraph, we conclude that $\alpha^{n_k}\pi_i^{n_k} + \beta^{n_k}\pi_j^{n_k}\equiv \phi(\p^{n_k})\rightarrow \phi(\p) \equiv \alpha\pi_i + \beta\pi_j$.
		Since $\pi_i\neq\pi_j$, this implies that $(\alpha^{n_k},\beta^{n_k})$ converges to $(\alpha,\beta)$.
		Thus $(\alpha^n,\beta^n)$ converges to $(\alpha,\beta)$.
	\end{step}

	\begin{step}\label{step:sigmapos}
		Now we deduce that $\sigma_i$ is always equal to 1.
		Assume for contradiction that there is a non-degenerate profile $\tilde\p\in\bar{\mathcal R}^I$ with $I_{\p} = \{i,j\}$ and 
		$\sigma_i(\tilde\p) = -1$.
		Since $\sigma_i\nu_i$ and $\sigma_j\nu_j$ are continuous at $\tilde\p$ by \Cref{step:lambdacont}, we can find a neighborhood of $\tilde\p$ such that $\sigma_i(\p) = -1$ for all non-degenerate profiles $\p$ contained in it.
		In particular, we can find $\epsilon > 0$ such that $\sigma_i(\p) = -1$ whenever $I_{\p} = \{i,j\}$, $\s[i] = \tilde{\s[i]}$, $u_j = \tilde u_j$, and $\sup\{|\pi_j(E) - \tilde\pi_j(E)|\colon E\in\mathcal E\} < \epsilon$.
		Let $\tilde{\mathcal P}$ be the corresponding set of profiles with concerned agents $\{i,j\}$.
		By \citeauthorhloff{Liap40a}'s theorem, we can find an event $E$ such that $\tilde\pi_i(E) = \tilde\pi_j(E) = \frac{\epsilon}2$. 
		Then $\tilde{\mathcal P}$ contains a non-degenerate profile $\p$ with $\sigma_i(\p) =-1$, $\pi_i(E) = \tilde\pi_i(E) = \frac{\epsilon}2$ and $\pi_j(E) = 0$.
		This is not possible, since $\sigma_i(\p)\nu_i(\s[i])\pi_i(E) + \sigma_j(\p)\nu_j(\s[j])\pi_j(E)$ would be negative.

		Since $j$ was arbitrary and $\sigma_i^J,\sigma_j^J$ satisfy the ratio condition stated in \Cref{lem:beliefindependence}, it follows that $\sigma_i^J$ is constant at 1 for all $J$.
		Since we have shown that $\sigma_i\nu_i$ is continuous, so is $\nu_i$.
	\end{step}
	\begin{step}
		Let $\bar\Phi$ be the SWF where $\bar\Phi(\p)$ is represented by $\bar\phi(\p) \equiv \sum_{i\in I_{\p}} \nu_i(\s[i])\pi_i$ and $\bar\psi(\p) \equiv \sum_{i\in I_{\p}} \omega_i(\s[i])u_i$ for every profile $\p$.
		We know from \Cref{lem:beliefindependence} and \Cref{lem:utilityindependence} that $\Phi$ and $\bar\Phi$ agree on all non-degenerate profiles.
		Our task is to show that they agree on every profile $\p\in\bar{\mathcal R}^I$ with non-degenerate utilities.
		
		\Cref{lem:utilityindependence} applies to all profiles with non-degenerate utilities. 
		Thus, $\psi(\p) = 0$ if and only if $\bar\psi(\p)\equiv 0$, and so $\Phi(\p)$ is complete indifference if and only if $\bar\Phi$ is.
		If $\Phi(\p)$ is complete indifference, $\phi(\p)$ may be arbitrary and choosing $\phi(\p) = \bar\phi(\p)$ is no restriction.

		If $\Phi(\p)$ and $\bar\Phi(\p)$ are not complete indifference, $\psi(\p) = \bar\psi(\p)$ follows from the fact that $\Pi$ is Hausdorff and $\psi$ and $\bar\psi$ are continuous and agree on a dense subset of $\bar{\mathcal R}^I$.
	\end{step}
\end{proof}

\begin{remark}[Dropping continuity in \Cref{prop:monotonicity}]
	Continuity is necessary to conclude that the weight of an agent's belief is always strictly positive.
	Consider $\Omega = [0,1]$ equipped with the Borel sigma-algebra $\mathcal B([0,1])$.
	Let $\tilde\pi$ be the uniform distribution on $\Omega$ and, for a non-atomic measure $\pi$ on $\Omega$, let $\rho(\pi) =\sup\{\frac{\pi(E)}{\tilde\pi(E)}\colon E\in\mathcal B([0,1])\}$.
	(Since non-atomic measures have continuous density functions, $\rho(\pi)$ is finite.)
	For $i\ge 2$, let $\nu_i(\s[i]) = \frac{1}{3^i\rho(\pi_i)}$, and $\nu_1$ as well as all $\omega_i$ be constant at 1.
	Suppose $\Phi(\p)$ is represented as in \Cref{prop:monotonicity} except that society's belief is $\pi_1 - \sum_{i\in I-\{1\}} \nu_i(\s[i])\pi_i$ (suitably scaled) whenever agent 1 is concerned and $\pi_1 = \tilde\pi$.
	Then $\Phi$ satisfies all axioms but continuity.
\end{remark}

\section{Implications of Independence of Redundant Actions}

Using independence of redundant acts, we derive a lemma which, together with \Cref{prop:monotonicity}, concludes the proof of \Cref{prop:main}.
But first we need two auxiliary statements.
Recall that a function is simple if it has finite range.

\begin{lemma}\label{lem:simple}
	Let $i\in I$ and $\p\in\bar{\mathcal R}^I$ such that $u_j$ is simple for all $j\in I-\{i\}$.
	Then for every act $f$, there is a simple act $g$ such that $f\sim_j g$ for all $j\in I$.
\end{lemma}

\begin{proof}
	Put differently, we want to show that for every act $f$, there is a simple act $g$ such that $(\ev_{\s[j]}(f))_{j\in I} = (\ev_{\s[j]}(g))_{j\in I}$.
	
	We first show that the sets $O^+ = \{x\in O\colon u_i(x) \ge \ev_{\s[i]}(f)\}$ and $O^- = \{x\in O\colon u_i(x) \le \ev_{\s[i]}(f)\}$ are non-empty.
	If $O^+$ is empty, then $\Omega = \bigcup_{k\in I} \{s\in\Omega\colon u_i(f(s)) \le \ev_{\s[i]}(f)-\frac1k\}$.
	Note that all sets in this union are measurable.
	Since $\pi_i$ is countably additive, there is $k_0$ such that $\pi_i(\{s\in\Omega\colon u_i(f(s)) \le \ev_{\s[i]}(f)-\frac1{k_0}\}) = \epsilon > 0$.
	The fact that $O^+$ is empty then gives $\ev_{\s[i]}(f) \le \ev_{\s[i]}(f) - \frac{\epsilon}{k_0}$, which is a contradiction.
	Similarly, one shows that $O^-$ is non-empty.
	
	Now let $V = \{\bm u_{-i}(x)\colon x\in O\}\subset \mathbb R^{I-\{i\}}$ be the range of $\bm u_{-i}$.
	Since all $u_j$ are simple, $V$ is finite.
	We partition the set of outcomes $O$ into the measurable sets $O_v = \bm u_{-i}^{-1}(v)$ and the set of states $\Omega$ into the measurable sets $E_v = f^{-1}(O_v)$ with $v$ ranging over $V$.
	To define $g$, consider two cases.
	If $\pi_i(E_v) = 0$, choose $x_v\in O_v$ arbitrarily and let $g(s) = x_v$ for all $s\in E_v$.
	If $\pi_i(E_v) > 0$, then $\frac1{\pi_i(E_v)}\pi_i|_{E_v}$ is a probability measure on $E_v$, where $\pi_i|_{E_v}$ is $\pi_i$ restricted to events contained in $E_v$. 
	By the previous paragraph, the sets $O_v^+ = \{x\in O_v\colon u_i(x) \ge \frac1{\pi_i(E_v)}\int_{E_v} (u_i\circ f)d\pi_i|_{E_v}\}$ and $O_v^- = \{x\in O_v\colon u_i(x) \le \frac1{\pi_i(E_v)}\int_{E_v} (u_i\circ f)d\pi_i|_{E_v}\}$ are non-empty.
	Thus, there are $x^+\in O_v^+$ and $x^-\in O_v^-$ and $\alpha\in[0,1]$ such that $\alpha u_i(x^+) + (1-\alpha) u_i(x^-) = \frac{1}{\pi_i(E_v)}\int_{E_v} (u_i\circ f)d\pi_i|_{E_v}$.
	Since $\pi_i$ is non-atomic, there is $E_v^+\subset E_v$ such that $\pi_i(E_v^+) = \alpha \pi_i(E_v)$.
	We define $g(s) = x^+$ for $s\in E_v^+$ and $g(s) = x^-$ for $s\in E_v - E_v^+$.
	This gives
	\begin{align*}
		\int_{E_v} (u_i\circ f)d\pi_i|_{E_v} &= \pi_i(E_v)\left(\alpha u_i(x^+) + (1-\alpha) u_i(x^-)\right) \\
		&= \pi_i(E_v^+) u_i(x^+) + \pi_i(E_v - E_v^+)u_i(x^-) = \int_{E_v} (u_i\circ g)d\pi_i|_{E_v}
	\end{align*}
	Also, since $\bm u_{-i}$ is constant on $O_v$, $\int_{E_v} (u_j\circ f)d\pi_j|_{E_v} = \int_{E_v} (u_j\circ g)d\pi_j|_{E_v}$ for all $j\in I-\{i\}$.
	In summary, we have $(\ev_{\s[j]}(f))_{j\in I} = (\ev_{\s[j]}(g))_{j\in I}$.
\end{proof}

\begin{lemma}\label{lem:subalgebra}
	Let $\pi$ and $\pi'$ be non-atomic probability measures on $(\Omega,\mathcal E)$.
	Then there is a sub-sigma-algebra $\mathcal E'$ of $\mathcal E$ such that $\pi|_{\mathcal E'} = \pi'|_{\mathcal E'}$ ($\pi(E) = \pi'(E)$ for all $E\in\mathcal E'$) and $\pi|_{\mathcal E'},\pi'|_{\mathcal E'}$ are non-atomic on $(\Omega,\mathcal E')$.
\end{lemma}

\begin{proof}
	We first construct an increasing sequence of sub-sigma-algebras $\mathcal E_k\subset\mathcal E$ so that $\pi$ and $\pi'$ agree on each $\mathcal E_k$.
	In the second step, we show that the sigma-algebra generated by $\bigcup_{k\ge 0} \mathcal E_k$ has the desired properties.

	We define $\mathcal E_k$ inductively so that $\mathcal E_k\subset\mathcal E_{k+1}$ and $\mathcal E_k$ is the sigma-algebra generated by a partition $\{E_k^1,\dots,E_k^{2^k}\}$ of $\Omega$ with $\pi(E_k^m) = \pi'(E_k^m) = \frac1{2^k}$ for all $m$.
	Let $\mathcal E_0 = \{\emptyset,\Omega\}$ (that is, the sigma-algebra generated by $\{\Omega\}$).
	Clearly, $\mathcal E_0$ has the required properties.
	Now let $k\ge 1$ and assume we have constructed $\mathcal E_l$ with the required properties for $l < k$.
	Let $\{E_{k-1}^1,\dots,E_{k-1}^{2^{k-1}}\}$ be the partition generating $\mathcal E_{k-1}$.
	For each $E_{k-1}^m$, we can use \citeauthor{Liap40a}'s theorem to divide it into two equal halves when measured by $\pi$ and $\pi'$.
	That is, we can find $E_k^{2m-1},E_k^{2m}\in\mathcal E$ such that $\{E_k^{2m-1},E_k^{2m}\}$ partitions $E_{k-1}^m$ and $\pi(E_k^{2m-1})= \pi'(E_k^{2m-1}) = \frac1{2^k}$.
	Then if we let $\mathcal E_k$ be the sigma-algebra generated by the partition $\{E_k^1,\dots,E_k^{2^k}\}$, it has the required properties.

	Now let $\mathcal F = \bigcup_{k\ge 0}\mathcal E_k$ and $\mathcal E'$ be the sigma-algebra generated by $\mathcal F$.
	We note two facts.
	First, $\mathcal F$ is an algebra.\footnote{A collection of subsets $\mathcal F$ of $\Omega$ is an algebra if (i) $\emptyset\in\mathcal F$, (ii) $\mathcal F$ is closed under taking complements, and (iii) $\mathcal F$ is closed under \emph{finite} unions.}
	Second, $\mathcal M = \{E\in\mathcal E\colon \pi(E) = \pi'(E)\}$ is a monotone class containing $\mathcal F$.\footnote{A collection of subsets $\mathcal M$ of $\Omega$ is a monotone class if it is closed under countable monotone unions and countable monotone intersections. 
	That is, $(E_k)_{k\ge 0}\subset\mathcal M$ with $E_0\subset E_1\subset\dots$ implies $\bigcup_{k\ge 0} E_k \in \mathcal M$ and $(E_k)_{k\ge 0}\subset\mathcal M$ with $E_0\supset E_1\supset\dots$ implies $\bigcap_{k\ge 0} E_k \in \mathcal M$.}
	Thus, the monotone class theorem implies that $\mathcal E'\subset\mathcal M$.
	
	It remains to show that $\pi|_{\mathcal E'}$ and $\pi'|_{\mathcal E'}$ are non-atomic on $(\Omega,\mathcal E')$.
	Since $\pi|_{\mathcal E'} = \pi'|_{\mathcal E'}$, it suffices to prove the statement for $\pi$.
	Let $E\in\mathcal E'$ with $\pi(E) > 0$.
	Choose $k$ such that $\frac1{2^k} < \pi(E)$.
	By construction of $\mathcal E_k$, we can partition $\Omega$ into sets $E^1,\dots,E^{2^k}\in\mathcal E_k\subset \mathcal E'$ such that $\pi(E^m) = \frac1{2^k}$ for all $m$.
	Note that $E\cap E^m\in\mathcal E'$ for all $m$ since $\mathcal E'$ is a sigma-algebra.
	Then the sets $E\cap E^1,\dots,E\cap E^{2^k}$ partition $E$ and so $\pi(E\cap E^m) > 0$ for some $m$.
	It follows that $0 < \pi(E\cap E^m) \le \pi(E^m) = \frac1{2^k} < \pi(E)$, which proves non-atomicity.
\end{proof}

\Cref{prop:monotonicity} characterizes SWFs that aggregate beliefs and utilities linearly so that the weight of each agent in either linear combination is a function of the agent's own preferences only.
The final lemma shows that for any such SWF $\Phi$ that additionally satisfies independence of redundant acts, the weights have to be constant.

\begin{lemma}\label{lem:lambdamu}
	For $i\in I$, let $\nu_i,\omega_i\colon\mathcal R\rightarrow\mathbb R_{++}$ be continuous functions; 
	let $\Phi$ be an SWF such that for every $\p\in\bar{\mathcal R}^{I}$, $\Phi(\p)$ is represented by $\frac1{\sum_{i\in I_{\p}} \nu_i(\s[i])}\sum_{i\in I_{\p}} \nu_i(\s[i])\pi_i$ and $\sum_{i\in I_{\p}}\omega_i(\s[i])u_i$.
	Then if $\Phi$ satisfies independence of redundant acts, $\nu_i$ and $\omega_i$ are constant for all $i$.
\end{lemma}

\begin{proof}
	\setcounter{step}{0}
	\setcounter{case}{0}
	Let $i,j\in I$ and $\s[i],\s[i]'\in\mathcal R$. 
	We want to show that $\nu_i(\s[i]) = \nu_i(\s[i]')$ and $\omega_i(\s[i]) = \omega_i(\s[i]')$.
	In the first step, we show that $\nu_i$ and $\omega_i$ are independent of $\pi_i$.
	In the rest of the proof, we show that they are independent of $u_i$, too.
\begin{step}\label{step:beliefindependent}
	Assume that $u_i = u_i'$.	
	First we show $\nu_i(\s[i]) = \nu_i(\s[i]')$.
	By \Cref{lem:subalgebra}, we can find a sub-sigma-algebra $\mathcal E'$ of $\mathcal E$ such that $\pi_i|_{\mathcal E'} = \pi_i'|_{\mathcal E'}$ and $\pi_i|_{\mathcal E'},\pi_i'|_{\mathcal E'}$ are non-atomic on $(\Omega,\mathcal E')$.
	We construct a belief for agent $j$ that allows us to leverage independence of redundant acts.
	Let $E\in\mathcal E'$ such that $\pi_i(E) = \frac12$. 
	If $\pi_i = \pi_i'$, there is nothing to show.
	Otherwise, either $\pi_i(F)\neq\pi'_i(F)$ for some $F\subset E$ or $\pi_i(F)\neq\pi'_i(F)$ for some $F\subset\Omega - E$.
	Assume the former is true.
	Then define $\pi_j$ so that $\pi_j(F) = 2\pi_i(F)$ for every $F\subset E$ and $\pi_j(\Omega - E) = 0$.
	Moreover, choose $u_j\in\mathcal U$ so that $u_j$ is simple and $u_j\not\equiv\pm u_i$, and let $\s[j]$ be represented by $\pi_j$ and $u_j$.
	
	The set of acts to which we will apply independence of redundant acts is $\mathcal A' = \mathcal A(\mathcal E',O)$.
	Let $\p$ be the profile where $i$ and $j$ have preferences $\s[i]$ and $\s[j]$ and all other agents are completely indifferent; let $\p' = (\p[-i],\s[i]')$.
	To meet the antecedent of independence of redundant acts, we have to show that $\mathcal A'$ is co-redundant for $\p$ and $\p'$.
	Condition~\ref{item:cored2} of co-redundancy is satisfied, since $\pi_i|_{\mathcal E'}$, $\pi_i'|_{\mathcal E'}$, and $\pi_j|_{\mathcal E'}$ are non-atomic by the choice of $\mathcal E'$.
	To verify condition~\ref{item:cored1}, we show that for every $f\in \mathcal A$, there is $g\in \mathcal A'$ such that $f\sim_i g$ and $f\sim_j g$.
	(The choice of $\mathcal A'$ and $u_i = u_i'$ ensure that also $f\sim_i' g$.)
	
	By \Cref{lem:simple}, we may assume that $f$ is simple.
	Define $g$ as follows: let $f(\Omega) = \{x_1,\dots,x_k\}$ be the range of $f$.
	For every $x_l$, let $\alpha_l = \pi_i(E\cap f^{-1}(x_l))$ and $\alpha_l^c = \pi_i((\Omega - E)\cap f^{-1}(x_l))$.
	(Note that $\pi_j(E\cap f^{-1}(x_l)) = 2\alpha_l$.)
	The non-atomicity of $\pi_i|_{\mathcal E'}$ allows us to find events $E_l\subset E$ and $E_l^c\subset \Omega - E$ in $\mathcal E'$ such that $\pi_i(E_l) = \alpha_l$ and $\pi_i(E_l^c) = \alpha_l^c$.
	In fact, we can partition $E$ and $\Omega - E$ into $\{E_1,\dots,E_k\}$ and $\{E_1^c,\dots,E_k^c\}$, respectively. 
	Then let $g(s) = x_l$ for $s\in E_l\cup E_l^c$.
	One can check that $\pi_j(E_l\cup E_l^c) = 2\pi_i(E_l) = 2\alpha_l$.
	Thus, $\ev_{\s[i]}(f) = \ev_{\s[i]}(g)$ and $\ev_{\s[j]}(f) = \ev_{\s[j]}(g)$ and so $f\sim_i g$ and $f\sim_j g$.
	
	Let ${\s} = f(\p)$ and ${\s'} = f(\p')$ and $\pi, u,\pi',u'$ be corresponding beliefs and utility functions.
	Independence of redundant acts applied to $\p$ and $\p'$ gives $g\srel g'$ if and only if $g\srel' g'$ for all $g,g'\in\mathcal A'$.

	Assume that $\nu_i(\s[i])\neq\nu_i(\s[i]')$.
	First, since $u_j\not\equiv\pm u_i$ and $u \equiv \nu_i(u_i)u_i + \nu_j(u_j)u_j$, $\s$ cannot be complete indifference, and so we can find outcomes $x$ and $y$ such that $x\succ y$.
	Recall that $\pi_i(E) = \pi_i'(E) = \frac12$ and $\pi_j(E) = 1$.
	It follows that $\pi(E)\neq\pi'(E)$ and $\pi(E),\pi'(E)>\frac12$.
	So there is an event $E'\in\mathcal E'$ such that $E'\subset E$ and $\pi(E')\neq\pi'(E') = \frac12$.
	Thus, $xE'y$ and $yE'x$ are acts in $\mathcal A'$ but $xE'y\not\sim yE'x$ and $xE'y\sim'yE'x$.
	This contradicts independence of redundant acts and so $\nu_i(\s[i]) = \nu_i(\s[i]')$.
	
	Second, assume that $\omega_i(\s[i]) \neq \omega_i(\s[i]')$.
	Since $u_j\not\equiv\pm u_i$, it follows that $u\neq u'$ and we can a find simple lotteries $p$ and $q$ on $O$ such that $u(p) > u(q)$ but $u'(q)\ge u'(p)$.
	By the previous paragraph, $\nu_i(\s[i]) = \nu_i(\s[i]')$ and so $\pi|_{\mathcal E'} = \pi'|_{\mathcal E'}$.
	Moreover, $\pi|_{\mathcal E'}$ is non-atomic as a weighted mean of measures that are non-atomic on $\mathcal E'$.
	So we can find acts $g$ and $g'$ in $\mathcal A'$ with $g\circ\pi = g\circ \pi' = p$ and $g'\circ \pi = g'\circ\pi' = q$.	
	This gives $g\succ g'$ but $g'\srel' g$, which contradicts independence of redundant acts.
	We conclude that $\omega_i(\s[i]) = \omega_i(\s[i]')$.
\end{step}

\begin{step}\label{step:utilityindependent}
	By \Cref{step:beliefindependent}, we can view $\nu_i$ and $\omega_i$ as functions $\nu_i(u_i)$ and $\omega_i(u_i)$ of $u_i$.
	We show that both these functions are constant.
	
	Recall that $\mathcal U$ consists of utility functions which are normalized to the unit interval, that is, $\inf_x u(x) = 0$ and $\sup_x u(x) = 1$.
	Let $\mathcal U' = \{u\in\mathcal U\colon \text{there exist $x,y\in O$ with } u(x)=0\text{ and }u(y) = 1\}$ be those utility functions for which the infimum and the supremum are attained.
	Observe that the closure of $\mathcal U'$ is $\mathcal U$. 
	Thus, since $\nu_i$ and $\omega_i$ are continuous, it suffices to show that they are constant on $\mathcal U'$.
	This we do now.	

		Let $u_i\in\mathcal U'$; let $x_0,x_1\in O$ such that $u_i(x_0) = 0$ and $u_i(x_1) = 1$ and $x^*\in O-\{x_0,x_1\}$ be arbitrary; let $u_i'$ be such that $u_i'(x) = u_i(x)$ for $x\in\{x_0,x_1,x^*\}$.
		We show that $\nu_i(u_i) = \nu_i(u_i')$ and $\omega_i(u_i) = \omega_i(u_i')$.
		Since $|O|\ge 4$, repeated application of this statement gives the same conclusion for all $u_i'\in\mathcal U'$.
		
		Let $u_i''$ be such that
		\begin{align*}
			u_i''(x) = 
			\begin{cases}
				u_i(x^*) &\text{if $u_i(x) < u_i(x^*)$ and $u_i'(x) > u_i(x^*)$}\\
				u_i'(x)	&\text{if $u_i(x) < u_i(x^*)$ and $u_i'(x) \le u_i(x^*)$}\\
				u_i(x)	&\text{if $u_i(x) \ge u_i(x^*)$}
			\end{cases}
		\end{align*}
		Note that $u_i''(x) = u_i'(x) = u_i(x)$ for $x\in\{x_0,x_1,x^*\}$.
		We want to apply independence of redundant acts to profiles with utility functions $(u_i,u_j,0,\dots,0)$ and $(u_i'',u_j,0,\dots,0)$ and the set of acts $\mathcal A' = \mathcal A(\mathcal E,\{x_0,x_1,x^*\})$.
		This requires choosing $u_j$ appropriately.
		Let $u_j$ be such that 
		\begin{align*}
			u_j(x) = 
			\begin{cases}
				0 &\text{if $u_i(x)\le u_i(x^*)$}\\
				u_i(x)	&\text{otherwise}
			\end{cases}
		\end{align*}
		\Cref{fig:utilityspace} depicts the images of $(u_i,u_j,0,\dots,0)$ and $(u_i'',u_j,0,\dots,0)$ in utility space.
		From $u_i$ to $u_i''$, we adjust the utility for outcomes with $u_i(x) \le u_i(x^*)$ toward $u_i'(x)$ without raising it above $u_i(x^*)$.
		Setting $u_j$ as we did, we can now apply independence of redundant acts to the corresponding profiles.
		
		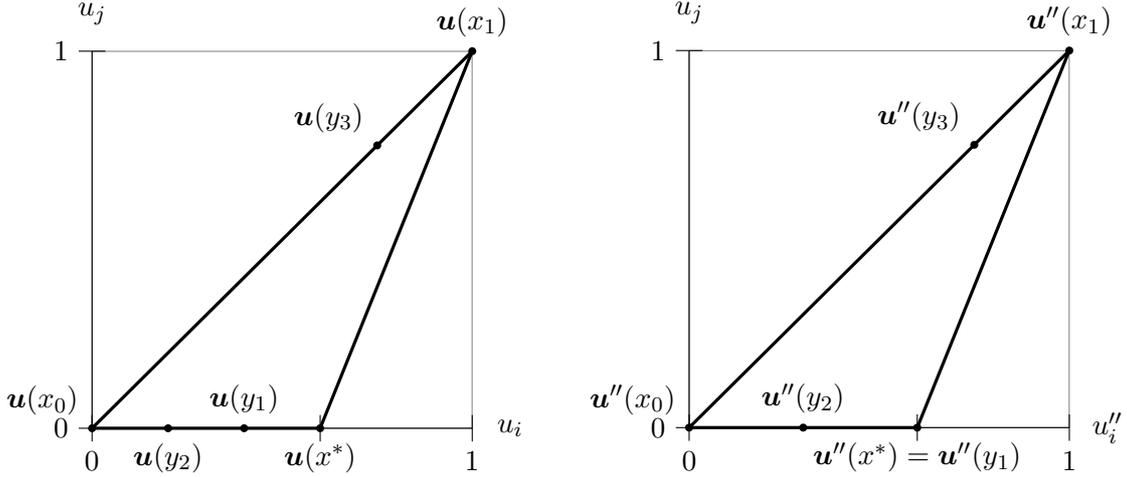
\begin{figure}
			\centering
			\begin{tikzpicture}[scale=5]
			  \draw[gray,very thin,step=1] (0,0) grid (1,1);
			 
			  \draw (0,0) -- (1,0);
			  \foreach \x/\xtext in {0/0, 0.6/, 1/1}
			    \draw[shift={(\x,0)}] (0pt,1pt) -- (0pt,-1pt) node[below] {$\xtext$};
				 \node at (1.1,0) {$u_i$};
			 
			  \draw (0,0) -- (0,1);
			  \foreach \y/\ytext in {0/0, 1/1}
			    \draw[shift={(0,\y)}] (1pt,0pt) -- (-1pt,0pt) node[left] {$\ytext$};
				 \node at (0,1.1) {$u_j$};
			 
			  \node[circle, fill=black, inner sep=0pt, minimum size=3pt, label = above left:$\bm u(x_0)$] at (0,0) {};
			  \node[circle, fill=black, inner sep=0pt, minimum size=3pt, label = above:$\bm u(x_1)$] at (1,1) {};
			  \node[circle, fill=black, inner sep=0pt, minimum size=3pt, label = below:$\bm u(x^*)$] at (.6,0) {};
			  \node[circle, fill=black, inner sep=0pt, minimum size=3pt, label = below:$\bm u(y_2)$] at (.2,0) {};
			  \node[circle, fill=black, inner sep=0pt, minimum size=3pt, label = above:$\bm u(y_1)$] at (.4,0) {};
			  \node[circle, fill=black, inner sep=0pt, minimum size=3pt, label = above left:$\bm u(y_3)$] at (.75,.75) {};

			  \draw[shorten >=-.15pt,very thick,triangle 90 cap-triangle 90 cap] (0,0) -- (1,1);
			  \draw[shorten >=-.15pt,very thick] (0,0) -- (.6,0);
			  \draw[shorten >=-.15pt,very thick,triangle 90 cap-triangle 90 cap] (0.6,0) -- (1,1);
			\end{tikzpicture}
			\hfil
			\begin{tikzpicture}[scale=5]
			  \draw[gray,very thin,step=1] (0,0) grid (1,1);
			 
			  \draw (0,0) -- (1,0);
			  \foreach \x/\xtext in {0/0, 0.6/, 1/1}
			    \draw[shift={(\x,0)}] (0pt,1pt) -- (0pt,-1pt) node[below] {$\xtext$};
				 \node at (1.1,0) {$u_i''$};
			 
			  \draw (0,0) -- (0,1);
			  \foreach \y/\ytext in {0/0, 1/1}
			    \draw[shift={(0,\y)}] (1pt,0pt) -- (-1pt,0pt) node[left] {$\ytext$};
				 \node at (0,1.1) {$u_j$};
			 
			  \node[circle, fill=black, inner sep=0pt, minimum size=3pt, label = above left:$\bm u''(x_0)$] at (0,0) {};
			  \node[circle, fill=black, inner sep=0pt, minimum size=3pt, label = above:$\bm u''(x_1)$] at (1,1) {};
			  \node[circle, fill=black, inner sep=0pt, minimum size=3pt, label = below:{$\bm u''(x^*) = \bm u''(y_1)$}] at (.6,0) {};
			  \node[circle, fill=black, inner sep=0pt, minimum size=3pt, label = above:$\bm u''(y_2)$] at (.3,0) {};
			  \node[circle, fill=black, inner sep=0pt, minimum size=3pt, label = above left:$\bm u''(y_3)$] at (.75,.75) {};

			  \draw[shorten >=-.15pt,very thick,triangle 90 cap-triangle 90 cap] (0,0) -- (1,1);
			  \draw[shorten >=-.15pt,very thick] (0,0) -- (.6,0);
			  \draw[shorten >=-.15pt,very thick,triangle 90 cap-triangle 90 cap] (0.6,0) -- (1,1);
			\end{tikzpicture}
			
			\caption{The images of $\bm u = (u_i,u_j,0,\dots,0)$ and $\bm u'' = (u_i'',u_j,0,\dots,0)$ in utility space projected onto the coordinates $\{i,j\}$. For example, $\bm u(x_0) = (u_i(x_0),u_j(x_0),0,\dots,0) = (0,0,0,\dots,0)$. 
			The outcomes $y_1,y_2$, and $y_3$ are examples for the three cases in the definition of $u_i''$.}
			\label{fig:utilityspace}
		\end{figure}
		
		Let $\pi_i,\pi_j\in\Pi$ with $\pi_i\neq\pi_j$ and $\s[i],\s[i]''$, and $\s[j]$ be represented by the pairs $(\pi_i,u_i), (\pi_i,u_i'')$, and $(\pi_j,u_j)$, respectively.
		Let $\p$ be the profile where $i$ and $j$ have preferences $\s[i]$ and $\s[j]$ and all other agents are completely indifferent; let $\p'' = (\p[-i],\s[i]'')$.
		First, since $u_i(x) = u_i''(x)$ for $x\in\{x_0,x_1,x^*\}$, it is clear that $\p$ and $\p''$ agree on the preferences over acts in $\mathcal A'$.
		Second, since $\bm u(x)$ is in the convex hull of $\{\bm u(x_0), \bm u(x_1),\bm u(x^*)\}$ for all $x\in O$, we have that for every act $f\in\mathcal A$, there is an act $g\in\mathcal A'$ such that $f\sim_i g$ and $f\sim_j g$.
		The analogous assertion holds for $\s[i]''$ and $\s[j]$. 
		Thus, $\mathcal A'$ satisfies condition~\ref{item:cored1} of co-redundancy for $\p$ and $\p''$.
		Condition~\ref{item:cored2} holds since $\pi_i$ and $\pi_j$ are non-atomic on $\mathcal E$.
		It follows from independence of redundant acts that with $\s = \Phi(\p)$ and $\s'' = \Phi(\p'')$, we have for all $g,g'\in\mathcal A'$, $g\srel g'$ if and only if $g\srel'' g'$.
		Let $(\pi,u)$ and $(\pi'',u'')$ be the beliefs and utility functions associated with $\s$ and $\s''$, respectively.
		Note that $u(x_0) = u''(x_0) = 0$ and $u(x_1) = u''(x_1) = 1$.
		
		If $\nu_i(u_i)\neq \nu_i(u_i'')$, then $\pi\neq\pi''$ since $\pi_i\neq\pi_j$.
		So we can find an event $E$ such that $\pi(E) = \frac12 \neq \pi''(E)$.
		It follows that $x_0Ex_1\sim x_1Ex_0$ but $x_0Ex_1\not\sim'' x_1Ex_0$, which is a contradiction since both acts are in $\mathcal A'$.

		If $\omega_i(u_i)\neq \omega_i(u_i'')$, then $u(x^*)\neq u''(x^*)$, since $u_i(x^*) = u_i''(x^*)\neq u_j(x^*)$.
		Let $E$ be an event such that $\pi(E) = \pi''(E) = u(x^*)$.
		Then $x^*\sim x_1Ex_0$ but $x^*\not\sim x_1Ex_0$, which is again a contradiction.
		
		We conclude that $\nu_i(u_i) = \nu_i(u_i'')$ and $\omega_i(u_i) = \omega_i(u_i'')$.
		The function $u_i''$ is ``closer'' to $u_i'$ than is $u_i$, since we have constructed it by moving utilities toward those in $u_i'$.
		Two more modifications of agent 1's utility function along the same lines will result in $u_i'$.
		To this end, we apply the same construction first to the profiles with utility functions $(u_i'',u_j',0,\dots,0)$ and $(u_i''',u_j',0,\dots,0)$ and then to profiles with utility functions $(u_i''',u_j,0,\dots,0)$ and $(u_i',u_j,0,\dots,0)$ (and the same beliefs $\pi_i$ and $\pi_j$).
		\begin{align*}
			u_i'''(x) = 
			\begin{cases}
				u_i''(x^*) &\text{if $u_i''(x) \ge u_i''(x^*)$ and $u_i'(x) < u_i''(x^*)$}\\
				u_i'(x)	&\text{if $u_i''(x) \ge u_i''(x^*)$ and $u_i'(x) \ge u_i''(x^*)$}\\
				u_i''(x)	&\text{if $u_i''(x) < u_i''(x^*)$}
			\end{cases}
			&&
			u_j'(x) = 
			\begin{cases}
				1 &\text{if $u_i''(x)\ge u_i''(x^*)$}\\
				u_i''(x)	&\text{otherwise}
			\end{cases}
		\end{align*}
	In summary, this gives $\nu_i(u_i) = \nu_i(u_i')$ and $\omega_i(u_i) = \omega_i(u_i')$ and proves the lemma.
	\end{step}
\end{proof}

We complete the proof of \Cref{prop:main}.

\main*

\begin{proof}
		First we prove that every SWF satisfying the axioms admits the desired representation. 
		It follows from \Cref{prop:monotonicity} and \Cref{lem:lambdamu} that there are $\bm v,\bm w \in \mathbb R^I_{++}$ such that $\Phi(\p)$ is represented by $\bar\phi(\p) = \frac{1}{\sum_{I_{\p}} v_i}\sum_{I_{\p}}v_i\pi_i$ and $\bar\psi(\p) = \sum_{I_{\p}} w_i u_i$ for all profiles $\p$ with non-degenerate utilities.
		Let $\bar\Phi$ be the SWF represented by $\bar\phi$ and $\bar\psi$.
		We want to show that $\Phi$ and $\bar\Phi$ agree on every non-common utility profile $\p$.
		\begin{case}\label{case:nonindiff}
			Suppose neither $\Phi(\p)$ nor $\bar\Phi(\p)$ is complete indifference. 
			Then it follows from the fact that that $\mathcal R$ is Hausdorff and $\Phi$ and $\bar\Phi$ are continuous and agree on the dense subset of $\bar{\mathcal R}^I$ of profiles with non-degenerate utilities that $\Phi(\p) = \bar\Phi(\p)$.
		\end{case}
		\begin{case}\label{case:dim1}
			Assume that $\emptyset \neq I_{\p}\neq I$ and there is $u\in\mathcal U$ such that $u_i \equiv\pm u$ for all $i\in I_{\p}$. 

		  Suppose $\Phi(\p)$ is complete indifference and $\bar\Phi(\p)$ is not complete indifference.
		  Let $i \in I - I_{\p}$ and $\s[i]' \in \mathcal R$ so that $u_i' \not\equiv \pm u$.
		  \Cref{lem:utilitylinear1} implies that $\psi(\p[-i],\s[i]') = u_i'$.
		  Moreover, by definition, $\bar\psi(\p[-i],\s[i]') \neq u_i'$ and $\bar\Phi(\p[-i],\s[i]')$ is not complete indifference.
		  In particular, $\Phi(\p[-i],\s[i]') \neq \bar\Phi(\p[-i],\s[i]')$.
		  This contradicts \Cref{case:nonindiff}.
			
			Suppose $\Phi(\p)$ is not complete indifference and $\bar\Phi(\p)$ is complete indifference.
			It follows from the definition of $\bar\Phi$ that $\sum_{i\in I_{\p}} w_iu_i \equiv 0$.
			Let $\underbar o,o_1,o_2,\bar o \in O$ with $u(\underbar o) \le u(o_1) \le u(o_2) \le u(\bar o)$ and $u(\underbar o) < u(\bar o)$.
			Let $(\p^n)$ be a sequence of profiles with non-degenerate utilities such that $I_{\p^n} = I_{\p}$ and $u_i^n(o) = u_i(o)$ for all $o \in O - \{o_1,o_2\}$ and $u^n(\underbar o) \le u^n(o_1) \le u^n(o_2) \le u^n(\bar o)$ for all $i \in I_{\p}$ and $n \in \mathbb N$. 
			(Such profiles exist since the space of utility functions satisfying the preceding conditions has dimension 3.)
			 
			 By construction, we have $\bar\psi(\p^n)(\underbar o) = \bar\psi(\p^n)(\bar o)$.
			 On the other hand, \Cref{lem:utilitylinear2} implies that $\psi(\p) \equiv \pm u$ so that $\psi(\p)(\underbar o) \neq \psi(\p)(\bar o)$.
			 Since $\psi$ is continuous by \Cref{lem:cont}, it follows that $\psi(\p^n)(\underbar o) \neq \psi(\p^n)(\bar o)$ for sufficiently large $n$.
			 This contradicts that $\Phi$ and $\bar\Phi$ agree on $\p^n$ for all $n$.\footnote{The argument for the case when $\Phi(\p)$ is complete indifference does not work here since \Cref{lem:utilitylinear1} does not rule out that $\psi(\p[-i],\s[i]') = u_i'$ even if $\psi(\p) \not \equiv \pm u$.}			  
		\end{case}
		\begin{case}
			Assume that $\p$ has utility dimension at least 2 and exactly one of $\Phi(\p)$ and $\bar\Phi(\p)$ is complete indifference.
			We may choose $\p$ so that $I_{\p}$ is minimal with these properties. 
			Note that $\Phi(\p[\sim i]) = \bar\Phi(\p[\sim i])$ for all $i \in I_{\p}$ by minimality and \Cref{case:dim1}.
			
			Suppose $\Phi(\p)$ is not complete indifference and $\bar\Phi(\p)$ is complete indifference.
			Let $i,j\in I_{\p}$ with $u_i \not\equiv \pm u_j$, which exist since $\p$ has utility dimension at least 2.
			By definition of $\bar\Phi$, $\bar\psi(\p[\sim i]) \equiv -u_i$ and $\bar\psi(\p[\sim j]) \equiv -u_j$. 
			Hence, $\psi(\p[\sim i]) \equiv -u_i$ and $\psi(\p[\sim j]) \equiv -u_j$.
			\Cref{lem:utilitylinear1} implies that $\psi(\p) \equiv \pm u_i$ and $\psi(\p) \equiv \pm u_j$.
			This contradicts $u_i\not\equiv\pm u_j$.
			
			Suppose $\Phi(\p)$ is complete indifference and $\bar\Phi(\p)$ is not complete indifference.
			\Cref{lem:utilitylinear1} implies that $\psi(\p[\sim i]) \equiv \pm u_i$ and $\psi(\p[\sim j]) \equiv \pm u_j$.
			 Hence, $\bar\psi(\p[\sim i]) \equiv \pm u_i$ and $\bar\psi(\p[\sim j]) \equiv \pm u_j$. 
			 The definition of $\bar\psi$ implies that $\bar\psi(\p) \equiv\pm u_i$ and $\bar\psi(\p)\equiv\pm u_j$.
			 This contradicts $u_i\not\equiv\pm u_j$.
		\end{case}
	\medskip
	
	Second, we prove that $\Phi$ satisfies the axioms on all profiles if it has a representation as in the statement of the theorem.
	It is easy to see that restricted monotonicity, faithfulness, no belief imposition, and continuity hold.
	We verify that independence of redundant acts is also satisfied.
	Let $\p,\p'\in\bar{\mathcal R}^I$ be two profiles.
	Let $\mathcal E'$ be a sub-sigma-algebra of $\mathcal E$ and $O'\subset O$ be a set of outcomes so that $\mathcal A' = \mathcal A(\mathcal E',O')$ is co-redundant for $\p$ and $\p'$.
	For all $i\in I$, both $(\pi_i|_{\mathcal E'},u_i|_{O'})$ and $(\pi_i'|_{\mathcal E'},u_i'|_{O'})$ represent $\s[i]|_{\mathcal A'} = \s[i]'|_{\mathcal A'}$.
	By condition~\ref{item:cored2} of co-redundancy, $\pi_i|_{\mathcal E'}$ and $\pi_i|_{\mathcal E'}$ are non-atomic.
	Hence, the representation of $\s[i]|_{\mathcal A'}$ is unique up to positive affine transformations of the utility function.
	That is, $\pi_i|_{\mathcal E'} = \pi_i'|_{\mathcal E'}$ and $u_i|_{O'}$ is a positive affine transformation of $u_i'|_{O'}$.
	By condition~\ref{item:cored1} of co-redundancy, $\inf\{u_i(x)\colon x\in O'\} = 0$ and $\sup\{u_i(x)\colon x\in O'\} = 1$, and likewise for $u_i'$.
	Thus, $u_i|_{O'} = u_i'|_{O'}$.
	It follows that $\sum_{i\in I} v_i\pi_i|_{\mathcal E'} = \sum_{i\in I} v_i\pi_i'|_{\mathcal E'}$ and $\sum_{i\in I} w_i u_i|_{O'} = \sum_{i\in I} w_i u_i'|_{O'}$, and so $\Phi(\p)|_{\mathcal A'} = \Phi(\p')|_{\mathcal A'}$.
\end{proof}

\cormain*

\begin{proof}
	\setcounter{step}{0}
	\setcounter{case}{0}
	We can assume that $\Phi$ has the representation as in \Cref{prop:main} for all non-common utility profiles.
	We show that all $v_i$ and all $w_i$ have to be equal. 
	
	Consider two agents $i,j\in I$ and any profile $\p\in\bar{\mathcal R}^I$ with $I_{\p} =\{i,j\}$ so that $\pi_i\neq\pi_j$ and $u_i\neq u_j$.
	Let $\p'$ be the profile obtained from $\p$ when $i$ and $j$ swap their preferences. 
	That is, $\s[i]' = \s[j]$ and $\s[j]' = \s[i]$.	
	Anonymity requires that $\Phi(\p) = \Phi(\p')$.
	Hence, $v_i\pi_i +  v_j \pi_j = v_i\pi_j + v_j \pi_i$ and $w_iu_i + w_ju_j \equiv w_iu_j + w_ju_i$.
	These equalities can hold only if $v_i = v_j$ and $w_i = w_j$.
	Since multiplication of all weights by the same positive constant does not change the preferences of society, we may assume that all weights are equal to 1.
\end{proof}

\begin{remark}[Proof for identical beliefs]\label{rem:identicalbeliefsproof}
	Let $\pi\in\Pi$ and denote by $\bar{\mathcal R}_\pi$ the set of all preference relations $\s$ for which there is $u\in \bar{\mathcal U}$ so that $\pi$ and $u$ represent $\s$. 
	From the proof of \Cref{prop:main}, one can obtain the following result.
	\begin{restatable}{proposition}{identicalbelief}\label{prop:identicalbelief}
		Let $\Phi$ be an SWF satisfying restricted monotonicity, independence of redundant acts, faithfulness, no belief imposition, and continuity on the domain $\bar{\mathcal R}_\pi^I$. 
		Then, there is $\bm w \in\mathbb R_{++}^I$ such that $\Phi(\p)$ is represented by $\pi$ and $\sum_{I_{\p}} w_i u_i$
		for all non-common utility profiles $\p\in\bar{\mathcal R}_\pi^I$.
	\end{restatable}
	Restricted monotonicity and faithfulness imply the restricted Pareto condition.
	Hence, the belief representing $\Phi$ is $\pi$ is for all $\p\in\bar{\mathcal R}_\pi^I$.
	\Cref{sec:beliefaggregation} can otherwise be omitted.
	\Cref{sec:tasteaggregation} stays the same.
	\Cref{step:lambda} and \Cref{step:sigmapos} of \Cref{prop:monotonicity}, \Cref{lem:subalgebra}, and \Cref{step:beliefindependent} and parts of \Cref{step:utilityindependent} of \Cref{lem:lambdamu} can also be omitted.
\end{remark}

\newpage

\section{Reference Table}\label{sec:referencetable}
	 
	 \vspace{2ex}
	\centering
	\begin{tabularx}{\textwidth}{lllX}
		\toprule
		\multicolumn{2}{l}{Symbol [elements]} & Name & Mathematical object\\
		\midrule
		$\Omega$ & & states & set \\
		$\mathcal E$  & $E$ & events & sigma-algebra on $\Omega$\\
		$O$ & $x,y$ & outcomes & set\\
		$\mathcal A$ &  $f,g$ & acts &  measurable functions $\Omega\rightarrow O$\\
		$\Pi$ & $\pi$ & beliefs &  non-atomic probability measures on $(\Omega,\mathcal E)$\\
		$\mathcal U$ ($\bar{\mathcal U}$) & $u$ & utility functions &  measurable functions $O\rightarrow \mathbb R$ with inf 0 and sup 1 (plus the function constant at 0)\\
		$\mathcal R$ ($\bar{\mathcal R}$) & $\s$ & preference relations & SEU preferences on $\mathcal A$ (plus complete indifference)\\
		$I$ & $i$ & agents & finite set\\
		$\bar{\mathcal R}^I$ & $\p$ & preference profiles & $I$-tuples with components in $\bar{\mathcal R}$\\
		$\s[i]$ & & preferences of agent $i$ & component of $\p$ with index $i$\\
		$\pi_i$ & & belief of agent $i$ & element of $\Pi$ representing $\s[i]$\\
		$u_i$ & & utility function of agent $i$ & element of $\bar{\mathcal U}$ representing $\s[i]$\\
		$\Phi$ & & social welfare function & function from $\bar{\mathcal R}^I \rightarrow \bar{\mathcal R}$\\
		\bottomrule
	\end{tabularx}

\end{document}